\def\draft{0}
\def\conferenceversion{0}

\documentclass[11pt]{article}

\usepackage{etex}
\usepackage{verbatim}
\usepackage{xspace,enumerate}
\usepackage[dvipsnames]{xcolor}
\usepackage[T1]{fontenc}
\usepackage[full]{textcomp}
\usepackage[american]{babel}
\usepackage{mathtools}
\usepackage{amsthm}
\usepackage[normalem]{ulem}
\usepackage{fullpage}
\usepackage{tcolorbox}
\usepackage{todonotes}

\usepackage{mathpazo}
\usepackage{multicol}
\usepackage{graphicx}

\usepackage[
letterpaper,
top=1in,
bottom=1in,
left=1in,
right=1in]{geometry}
\usepackage{newpxtext} %
\usepackage{textcomp} %
\usepackage[varg,bigdelims]{newpxmath}
\usepackage[scr=rsfso]{mathalfa}%
\usepackage{bm} %
\let\mathbb\varmathbb
\usepackage{microtype}
\usepackage[pagebackref,bookmarksnumbered,colorlinks=true,urlcolor=blue,linkcolor=blue,citecolor=OliveGreen]{hyperref}
\usepackage[capitalise,nameinlink]{cleveref}
\crefname{lemma}{Lemma}{Lemmas}
\crefname{fact}{Fact}{Facts}
\crefname{theorem}{Theorem}{Theorems}
\crefname{corollary}{Corollary}{Corollaries}
\crefname{claim}{Claim}{Claims}
\crefname{example}{Example}{Examples}
\crefname{algorithm}{Algorithm}{Algorithms}
\crefname{problem}{Problem}{Problems}
\crefname{definition}{Definition}{Definitions}
\crefname{exercise}{Exercise}{Exercises}
\usepackage{amsthm}

\newtheorem{theorem}{Theorem}[section]
\newtheorem*{theorem*}{Theorem}
\newtheorem{lemma}[theorem]{Lemma}
\newtheorem*{lemma*}{Lemma}
\newtheorem{fact}[theorem]{Fact}
\newtheorem*{fact*}{Fact}
\newtheorem{proposition}[theorem]{Proposition}
\newtheorem*{proposition*}{Proposition}

\newtheorem*{corollary*}{Corollary}

\newtheorem*{hypothesis*}{Hypothesis}
\newtheorem{conjecture}[theorem]{Conjecture}
\newtheorem*{conjecture*}{Conjecture}
\theoremstyle{definition}
\newtheorem{definition}[theorem]{Definition}
\newtheorem*{definition*}{Definition}

\newtheorem*{construction*}{Construction}

\newtheorem*{example*}{Example}
\newtheorem{question}[theorem]{Question}
\newtheorem*{question*}{Question}
\newtheorem{algorithm}[theorem]{Algorithm}
\newtheorem*{algorithm*}{Algorithm}

\newtheorem*{assumption*}{Assumption}

\newtheorem*{problem*}{Problem}

\newtheorem*{openquestion*}{Open Question}
\theoremstyle{remark}
\newtheorem{claim}[theorem]{Claim}
\newtheorem*{claim*}{Claim}
\newtheorem{remark}[theorem]{Remark}
\newtheorem*{remark*}{Remark}

\newtheorem*{observation*}{Observation}
\usepackage{paralist}
\let\originalleft\left
\let\originalright\right
\renewcommand{\left}{\mathopen{}\mathclose\bgroup\originalleft}
\renewcommand{\right}{\aftergroup\egroup\originalright}
\usepackage{turnstile}
\usepackage{mdframed}
\usepackage{tikz}
\usetikzlibrary{positioning}
\usepackage{caption}
\DeclareCaptionType{Algorithm}
\usepackage{newfloat}
\usepackage{array}
\usepackage{subfig}
\usepackage{bbm}
\usepackage{xparse}
\usepackage{amsthm} %
\makeatletter
\let\latexparagraph\paragraph
\RenewDocumentCommand{\paragraph}{som}{%
  \IfBooleanTF{#1}
    {\latexparagraph*{#3}}
    {\IfNoValueTF{#2}
       {\latexparagraph{\maybe@addperiod{#3}}}
       {\latexparagraph[#2]{\maybe@addperiod{#3}}}%
  }%
}
\newcommand{\maybe@addperiod}[1]{%
  #1\@addpunct{.}%
}
\makeatother

\usepackage{boxedminipage}

\newcommand{\Paren}[1]{\left(#1\right)}

\newcommand{\Brac}[1]{\left[#1\right]}

\newcommand{\abs}[1]{\lvert#1\rvert}

\newcommand{\card}[1]{\lvert#1\rvert}
\newcommand{\Card}[1]{\left\lvert#1\right\rvert}

\newcommand{\set}[1]{\{#1\}}

\newcommand{\norm}[1]{\lVert#1\rVert}
\newcommand{\Norm}[1]{\left\lVert#1\right\rVert}

\newcommand{\iprod}[1]{\langle#1\rangle}

\let\ip\Iprod

\newcommand{\Esymb}{\mathbb{E}}
\newcommand{\Psymb}{\mathbb{P}}

\DeclareMathOperator*{\E}{\Esymb}

\DeclareMathOperator*{\ProbOp}{\Psymb}
\renewcommand{\Pr}{\ProbOp}

\newcommand{\sge}{\succeq}

\renewcommand{\ij}{{ij}}

\newcommand\bdot\bullet

\DeclareMathOperator{\Tr}{Tr}

\DeclareMathOperator{\vol}{vol}

\DeclareMathOperator{\rank}{rank}

\newcommand{\ie}{i.e.,\xspace}

\newcommand{\etal}{et al.\xspace}

\newcommand{\Z}{\mathbb Z}

\newcommand{\N}{\mathbb N}
\newcommand{\R}{\mathbb R}
\newcommand{\C}{\mathbb C}

\newcommand{\problemmacro}[1]{\texorpdfstring{\textup{\textsc{#1}}}{#1}\xspace}

\newcommand{\uniquegames}{\problemmacro{unique games}}

\newcommand{\balancedseparator}{\problemmacro{balanced separator}}

\newcommand{\sparsestcut}{\problemmacro{sparsest cut}}

\newcommand{\edgeexpansion}{\problemmacro{edge expansion}}

\newcommand{\conductance}{\problemmacro{conductance}}
\newcommand{\cutimprovement}{\problemmacro{cut improvement}}

\newcommand{\smallsetexpansion}{\problemmacro{small-set expansion}}

\newcommand{\cN}{\mathcal N}

\newcommand{\cP}{\mathcal P}

\newcommand{\cS}{\mathcal S}

\newcommand{\bbP}{\mathbb P}

\renewcommand{\leq}{\leqslant}
\renewcommand{\le}{\leqslant}
\renewcommand{\geq}{\geqslant}

\newcommand*{\vertbar}{\rule[-1ex]{0.5pt}{2.5ex}}
\let\epsilon=\varepsilon
\numberwithin{equation}{section}
\newcommand\MYcurrentlabel{xxx}
\newcommand{\MYstore}[2]{%
  \global\expandafter \def \csname MYMEMORY #1 \endcsname{#2}%
}
\newcommand{\MYload}[1]{%
  \csname MYMEMORY #1 \endcsname%
}
\newcommand{\MYnewlabel}[1]{%
  \renewcommand\MYcurrentlabel{#1}%
  \MYoldlabel{#1}%
}
\newcommand{\MYdummylabel}[1]{}
\newcommand{\torestate}[1]{%
  \let\MYoldlabel\label%
  \let\label\MYnewlabel%
  #1%
  \MYstore{\MYcurrentlabel}{#1}%
  \let\label\MYoldlabel%
}
\newcommand{\restatedef}[1]{%
  \let\MYoldlabel\label
  \let\label\MYdummylabel
  \begin{definition*}[Restatement of \cref{#1}]
    \MYload{#1}
  \end{definition*}
  \let\label\MYoldlabel
}
\newcommand{\restatetheorem}[1]{%
  \let\MYoldlabel\label
  \let\label\MYdummylabel
  \begin{theorem*}[Restatement of \cref{#1}]
    \MYload{#1}
  \end{theorem*}
  \let\label\MYoldlabel
}
\newcommand{\restatelemma}[1]{%
  \let\MYoldlabel\label
  \let\label\MYdummylabel
  \begin{lemma*}[Restatement of \cref{#1}]
    \MYload{#1}
  \end{lemma*}
  \let\label\MYoldlabel
}
\newcommand{\restateprop}[1]{%
  \let\MYoldlabel\label
  \let\label\MYdummylabel
  \begin{proposition*}[Restatement of \cref{#1}]
    \MYload{#1}
  \end{proposition*}
  \let\label\MYoldlabel
}
\newcommand{\restatefact}[1]{%
  \let\MYoldlabel\label
  \let\label\MYdummylabel
  \begin{fact*}[Restatement of \cref{#1}]
    \MYload{#1}
  \end{fact*}
  \let\label\MYoldlabel
}
\newcommand{\restate}[1]{%
  \let\MYoldlabel\label
  \let\label\MYdummylabel
  \MYload{#1}
  \let\label\MYoldlabel
}

\newcommand{\eps}{\epsilon}

\allowdisplaybreaks
\newcommand*{\Id}{\mathbf{I}}

\newenvironment{algorithmbox}{\begin{mdframed}[nobreak=true]\begin{algorithm}}{\end{algorithm}\end{mdframed}}

\providecommand{\important}[1]{\texorpdfstring{\textup{\textsc{#1}}}{#1}\xspace}

\providecommand{\CD}{\important{SD}}
\providecommand{\CP}{\important{cp}}

\newcommand{\cay}{{\textup{\textsf{Cay}}}}
\let\Cay\cay

\definecolor{niceish}{HTML}{74b807} 
\ifnum\draft=1
\newcommand{\salil}[1]{\textcolor{blue}{[Salil: #1]}}
\newcommand{\jake}[1]{\textcolor{violet}{[Jake: #1]}}
\newcommand{\tom}[1]{\textcolor{WildStrawberry}{[Tommaso: #1]}}
\newcommand{\chris}[1]{\textcolor{niceish}{[Chris: #1]}}
\newcommand{\jiyu}[1]{\textcolor{Orange}{[Jiyu: #1]}}
\else
\newcommand{\salil}[1]{}
\newcommand{\jake}[1]{}
\newcommand{\tom}[1]{}
\newcommand{\chris}[1]{}
\newcommand{\jiyu}[1]{}
\fi

\newcommand{\el}{\ell}
\newcommand{\lam}{\lambda}

\newcommand{\al}{\alpha}

\newcommand{\gam}{\gamma}
\newcommand{\Gam}{\Gamma}

\newcommand{\om}{\om}

\newcommand{\matA}{\mathbf{A}}

\newcommand{\matD}{\mathbf{D}}

\newcommand{\matG}{\mathbf{G}}

\newcommand{\matL}{\mathbf{L}}

\newcommand{\matW}{\mathbf{W}}

\newcommand{\matPi}{\mathbf{\Pi}}

\DeclareMathOperator{\diam}{diam}

\DeclareMathOperator{\linspan}{span}

\RequirePackage[outline]{contour} 
\contourlength{0.065pt}
\contournumber{10}%

\newcommand{\low}{\textnormal{\textsc{low}}\xspace}

\newcommand{\mul}{\textnormal{\textsc{mul}}\xspace}

\ifpdf
\hypersetup{
    pdftitle={Sparsest cut and eigenvalue multiplicities on low degree Abelian Cayley graphs},
    pdfauthor={Tommaso d'Orsi, Chris Jones, Jake Ruotolo, Salil Vadhan, Jiyu Zhang}
}
\fi

\begin{document}
\date{}

\title{
Sparsest cut and eigenvalue multiplicities on low degree Abelian Cayley graphs 
\ifnum\draft=1 {\sc \small \\ Working Draft: Please Do Not Distribute} \fi
}

\ifnum\conferenceversion=1
\author{Author names omitted}
\else
\author{Tommaso d'Orsi\thanks{Bocconi University. \texttt{tommaso.dorsi@unibocconi.it}}
\and Chris Jones\thanks{Bocconi University. \texttt{chris.jones@unibocconi.it}}
\and Jake Ruotolo\thanks{Harvard University. \texttt{jakeruotolo@g.harvard.edu}}
\and Salil Vadhan\thanks{Harvard University. \texttt{salil\_vadhan@harvard.edu}}
\and Jiyu Zhang\thanks{Bocconi University. \texttt{jiyu.zhang@phd.unibocconi.it}}
}
\fi

\maketitle

\ifnum\conferenceversion=0
\centerline{\textit{In memory of Luca Trevisan (1971--2024)}}
\fi
\thispagestyle{empty}
\begin{abstract}
Whether or not the Sparsest Cut problem admits an efficient $O(1)$-approximation algorithm is a fundamental algorithmic question with connections to geometry and the Unique Games Conjecture.

Revisiting spectral algorithms for Sparsest Cut, we present a novel, simple algorithm that combines  eigenspace enumeration with a new algorithm for the Cut Improvement problem.
The runtime of our algorithm is parametrized by a quantity that we call the solution dimension $\CD_\varepsilon(G)$: the smallest $k$ such that the subspace spanned by the first $k$ Laplacian eigenvectors contains all but $\varepsilon$ fraction of a sparsest cut.

Our algorithm matches the guarantees of prior methods based on the threshold-rank paradigm, while also extending beyond them.
To illustrate this, we study its performance on low degree Cayley graphs over Abelian groups---canonical examples of graphs with poor expansion properties.

We prove that low degree Abelian Cayley graphs have small solution dimension,
yielding an algorithm that computes a $(1+\eps)$-approximation to the uniform Sparsest Cut of a degree-$d$ Cayley graph over an Abelian group of size $n$ in time $n^{O(1)}\cdot\exp\set{(d/\eps)^{O(d)}}$.

Along the way to bounding the solution dimension of Abelian Cayley graphs, we analyze their sparse cuts and spectra, proving that the collection of $O(1)$-approximate sparsest cuts has an $\eps$-net of size $\exp\{(d/\eps)^{O(d)}\}$ and

that the multiplicity of $\lam_2$ is bounded by $2^{O(d)}$.
The latter bound is tight and improves on a previous bound of $2^{O(d^2)}$ by Lee and Makarychev.

\end{abstract}
\clearpage
\thispagestyle{empty}
\microtypesetup{protrusion=false}
\tableofcontents{}
\microtypesetup{protrusion=true}

\clearpage
\pagestyle{plain}
\setcounter{page}{1}

\section{Introduction}\label{sec:introduction}

For an undirected  graph $G$, the sparsest cut measures how poor of an expander the graph is,
characterizing how slowly random walks on $G$ mix. For simplicity in this introduction, we will state our definitions just for the case of regular graphs:

\begin{definition}[Sparsest Cut]
    For a $d$-regular graph $G$ on $n$ vertices and a partition $(Q, \bar Q)$ of its vertex set $V(G)$, the 
    \emph{(normalized) density}
    of the cut $(Q, \bar Q)$ is defined as:
    \begin{align*}
        \psi_G(Q):=\frac{|E(Q, \bar{Q})|}{|Q|\cdot|\bar Q|}\cdot \frac{n}{d} \,.
    \end{align*}
    The \sparsestcut\ problem is, given $G$, to find $Q\subseteq V(G)$ that minimizes $\psi_G(Q).$  We write $\psi(G)$ to denote the minimum value of $\psi_G(Q)$ over all cuts $(Q, \bar Q)$.
\end{definition}
The normalization factor of $n/d$ ensures that the largest possible value of $\psi(G)$ is $1$, because for a randomly chosen cut $Q$, the expectation of $|E(Q,\bar Q)|$ is $nd/4$, and the expectation of $|Q|\cdot |\bar Q|$ is $n^2/4$.
We'll often switch back and forth between $\psi(G)$ and the closely related quantity of {\em conductance}, defined as $$\phi(G)=\min_{|Q|\leq n/2} \phi_G(Q),\quad \text{ where }\quad\phi_G(Q) = \frac{|E(Q,\bar Q)|}{d\cdot |Q|}.$$
We have $\phi(G)\leq \psi(G) \leq 2\cdot \phi(G)$, so the problem of approximating \conductance\ is equivalent to the problem of approximating \sparsestcut, up to multiplying the approximation factor by 2.  
Via known reductions, up to a constant factor of approximation,
\sparsestcut is also equivalent to approximating several other graph parameters such as \edgeexpansion and \balancedseparator \cite{arora2009expander}.

Because of its centrality to algorithms, computational complexity, combinatorics, and geometry, efficient algorithms for \sparsestcut\ have been the main focus of a long line of work. 
A polynomial-time $O(\log n)$-approximation algorithm was obtained via a linear programming relaxation by Leighton and Rao \cite{leighton1999multicommodity}, and later improved to $O(\sqrt{\log n})$ by Arora, Rao, and Vazirani (ARV) \cite{arora2009expander} via semidefinite programming. To this day, the ARV algorithm remains the state-of-the-art for general graphs. 

A central challenge to resolving the approximability of \sparsestcut
is its intricate relationship with the \emph{Unique Games Conjecture} (UGC) and the \textit{Small Set Expansion Hypothesis} (SSEH), themselves outstanding open problems with a close connection \cite{khot2002power, raghavendra2010graph, raghavendra2012reductions}. 
Assuming the SSEH, \sparsestcut does not have a polynomial-time $O(1)$-approximation algorithm \cite{raghavendra2012reductions} (the same holds for ``non-uniform'' \sparsestcut assuming the UGC \cite{chawla2006hardness, khot2015unique, arora2008unique}).
Thus we have some evidence that beating the $O(\sqrt{\log n})$ approximation factor of the ARV algorithm may be hard.
On the other side of the coin, searching for better approximation algorithms for \sparsestcut is an ostensible approach to developing algorithms for related problems, including \smallsetexpansion and \uniquegames.

Given the above situation, researchers have attempted to pair interesting instances with interesting algorithms to determine when constant-factor approximation to \sparsestcut is possible.
Below we summarize the most relevant algorithmic approaches and what is known about them.

\paragraph{Fiedler's Algorithm.}
The simplest algorithm for \sparsestcut is Fiedler's Algorithm, which is just to threshold the second eigenvector~\cite{fiedler1973algebraic}.
Fiedler's Algorithm always finds a cut $Q$ such that $\phi(G)\leq \sqrt{2\lambda_2}$, where $\lambda_2$ is the second-smallest eigenvalue of the normalized Laplacian of $G$.  Indeed, this performance guarantee for Fiedler's Algorithm is the proof of the right-hand side of Cheeger's Inequality, which says:
$$\lambda_2/2\leq \phi(G)\leq \sqrt{2\lambda_2}\,.$$
\salil{following sentence is new}\chris{this sentence felt a little misleading to me, if I understand correctly that Fiedler's algorithm actually works the best when the Cheeger lower bound is tight (so the approximation ratio is much better than $O(\sqrt{\lam_2}/\phi)$ in the cases where Fiedler is the most successful, e.g. the barbell graph). What do people think?}
This analysis of Fiedler's algorithm shows that it gives an $O(\sqrt{\lambda_2}/\phi(G))$ worst-case approximation ratio, which can be quite large.  Indeed, 
examples show that the approximation factor achieved by Fiedler's Algorithm can be as bad as $\Omega(n^{1/3})$ in general \cite{guattery1998quality}.

\paragraph{Expander-like graphs.} Starting around 15 years ago, a series of works gave approximation algorithms for   
\sparsestcut and \uniquegames 
on graphs that are similar to expanders.  A constant-factor approximation to \sparsestcut\ is trivial on actual expanders (every cut is a constant-factor approximation since $\psi(G)=\Omega(1)$), but it took a nontrivial work of Arora et al.~\cite{arora2008unique} to approximate \uniquegames\ on expanders.  This in turn led to efficient algorithms for \sparsestcut\ on graphs that are ``expander-like'' in the sense of having few small eigenvalues, as captured by the following definition:
\begin{definition}[$\tau$-threshold-rank]\label{def:tau-threshold-rank}
    For a graph $G$ and $2\geq \tau\geq 0\,,$ the  $\tau$-\emph{threshold-rank} $\mul_\tau(G)$ is the number of eigenvalues of the normalized Laplacian with value at most $\tau$.
\end{definition}
Recall that $G$ is an expander if and only if $\lambda_2$ is bounded away from 0.  Equivalently, there is a constant $\tau>0$ (independent of $n$) such that $\mul_\tau(G)=1$.  For every constant $\tau>0$, the works \cite{arora2015subexponential,barak2011rounding} give a constant-factor approximation algorithm for \sparsestcut\ that runs in time polynomial in $n$ but exponential in $\mul_\tau(G)$.
\salil{added BRS author names}
The algorithm of Barak, Raghavendra, and Steurer~\cite{barak2011rounding} is based on a rounding of a semidefinite programming relaxation, the $O(\mul_\tau(G))$'th level of the Sum-of-Squares hierarchy.
Recent work has extended the Sum-of-Squares technique based on looser expansion properties, including being a small-set expander \cite{bafna2021playing}, coming from a high-dimensional expander \cite{bafna2022high}, or having a ``succinct characterization'' of non-expanding sets \cite{bafna2023solving}.

The algorithm of Arora, Barak, and Steurer~\cite{arora2015subexponential} is based on enumerating the subspace spanned by the first $\mul_\tau(G)$ eigenvectors to find a vector that is very close to a sparse cut.
Approaches like this, which utilize multiple low eigenvectors of the Laplacian matrix, are natural generalizations of Fiedler's Algorithm, and we collectively refer to them as \emph{higher-order spectral algorithms}. \salil{minor edits to preceeding sentence}

\paragraph{Cheeger-lower-bound graphs.}
With a more careful use of the Sum-of-Squares hierarchy,
Guruswami and Sinop~\cite{guruswami2013approximating} showed that we can replace $\mul_{\Omega(1)}$ in the aforementioned results with $\mul_{O(\phi(G))}$.

\begin{theorem}[\cite{guruswami2013approximating}]\label{thm:threshold-rank}
    For all constants $\eps > 0\,$,
    \sparsestcut admits a $(1+\eps)$-approximation algorithm in time $n^{O(1)}\cdot \exp\{O(r)\}\,,$ where $r$ is the $O(\phi(G)/\eps)$-threshold-rank.
\end{theorem}

When $\phi(G)$ is small, it is not clear that graphs with small $O(\phi(G))$-threshold-rank are ``expander-like'' anymore.  Instead, we can think of them as graphs where the lower bound of Cheeger's Inequality ($\lambda_2/2\leq \phi(G)$) is nearly tight in the sense that $G$ has few eigenvalues between $\lam_2/2$ and $\phi(G)$.

The crucial feature of this class of graphs is that they have relatively few distinct near-sparsest cuts.
Indeed, the indicator vector of a set $Q$ with conductance $O(\phi(G))$ must be close to the subspace spanned by the eigenvalues of magnitude $O(\phi(G))$.
Thus if the latter subspace has dimension $r$, all of the near-sparsest cuts can be covered by a net of size exponential in $r$.

A similar consequence of small $O(\phi(G))$-threshold-rank is that these graphs can be decomposed into a small number
of pieces with conductance at most $O(\phi(G))$\,, such
as depicted in \cref{fig:three-communities}.
See \cref{app:threshold-rank} for a simple statement and proof of this decomposition.

\begin{figure}[ht]
    \centering
    \includegraphics[width=0.25\linewidth]{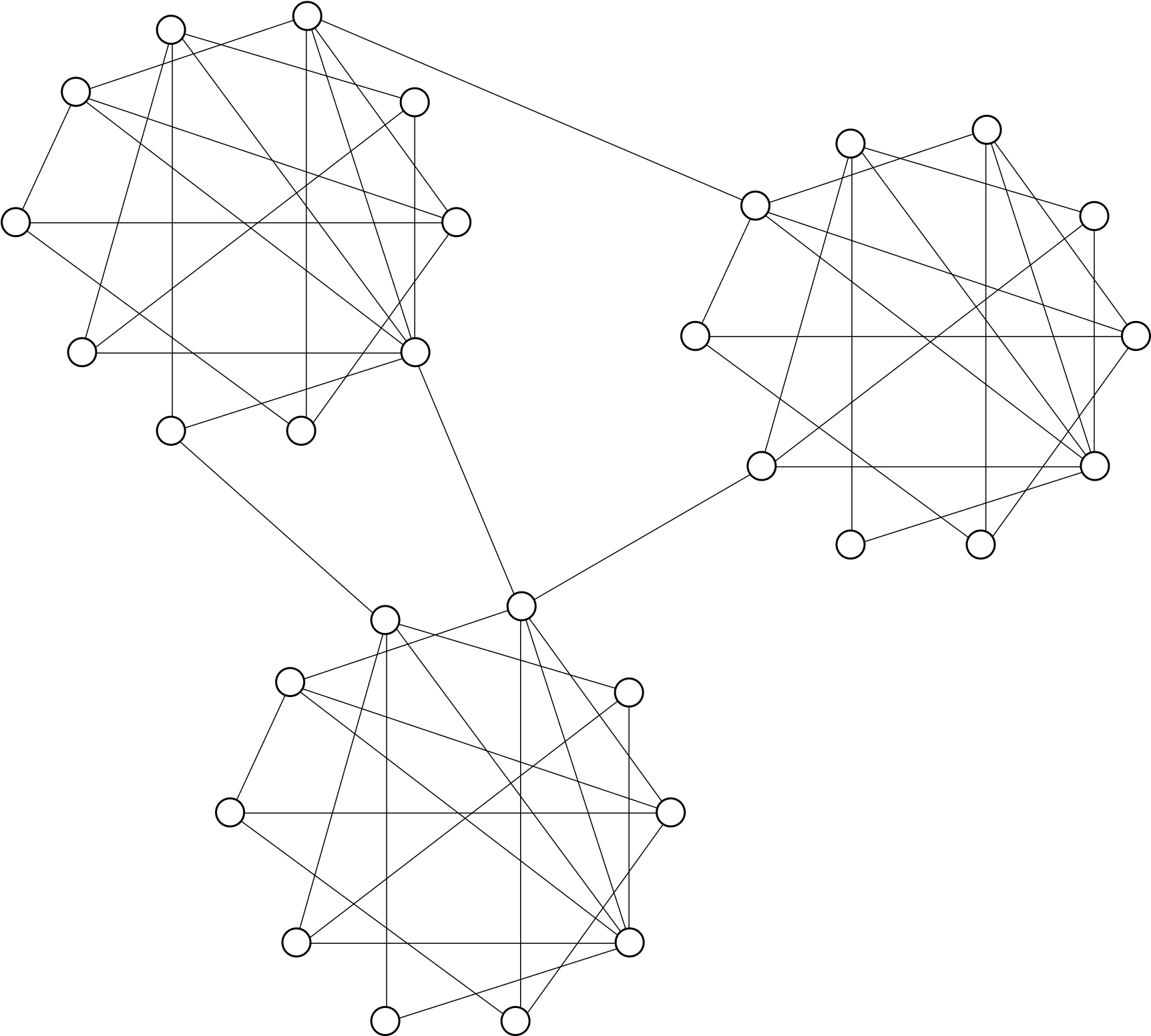}
    \caption{A graph with three distinct sparse cuts. In this graph, the 0/1 indicator vectors for the three pieces approximately span the three low eigenvectors. 
    }
    \label{fig:three-communities}
\end{figure}

\paragraph{Cut improvement.}
As noted by Guruswami and Sinop~\cite{guruswami2013approximating}, 
\cref{thm:threshold-rank} can also be proven via a higher-order spectral algorithm, based on eigenspace enumeration plus an algorithm for the \cutimprovement problem.\footnote{Guruswami and Sinop also study the {\em non-uniform} \sparsestcut\ problem, which is not covered by the results on \cutimprovement.}
Informally, \cutimprovement is the problem of: given a cut $Q \subseteq V(G)$, compute the sparsest cut that has large intersection with $Q$.\salil{changed `nearby in Hamming distance' to `has large intersection with'} 
This is a problem of independent interest, with practical applications~\cite{andersen2008algorithm,lang2009empirical,fountoulakis2023flow}.

Andersen and Lang \cite{andersen2008algorithm} gave a flow-based approximation algorithm for \cutimprovement\ 
which for our purposes can be summarized as follows.
\begin{theorem}[{\cite{andersen2008algorithm}}]
\label{thm:cut-improvement}
    Let $G$ be an $n$-vertex graph, let $0 < \eps < 1/2$, and let $Q^* \subseteq V(G)$\,.
    There is a polynomial-time algorithm that, given $G$ and $Q \subseteq V(G), |Q| \leq n/2$
    such that $|Q \cap Q^*| \geq (1 - \eps) |Q^*|$\,,
    returns $\hat Q \subseteq V(G)$ with $0 < \abs{\hat Q}\leq n/2$ and $\phi(\hat Q) \leq \frac{1}{1-2\eps} \cdot\phi(Q^*)$\,.

\end{theorem}
With this \cutimprovement\ algorithm in hand, we can derive \cref{thm:threshold-rank} (with a factor 2 loss in approximation factor) as follows: by
enumerating the eigenspace up to eigenvalue $O(\phi(G)/\eps)$, we can obtain a cut $Q$ that is $\eps$-close to a sparsest cut $Q^*$, and 
then we use \cref{thm:cut-improvement} to find $\hat Q$ such that $$\phi(\hat Q) \leq (1+O(\eps))\cdot \phi(Q^*) \leq (1+O(\eps))\cdot 2\psi(G),$$
where the factor 2 loss comes from moving between conductance and density.

\paragraph{Buser graphs.}
A class of graphs which is in a sense complementary to the class of Cheeger-lower-bound graphs consists of graphs that satisfy a \emph{Buser Inequality} $\phi = \Theta(\sqrt{\lam_2})$\,.
Examples of graphs with bounded degree $d$ which satisfy this inequality are: graphs with non-negative discrete curvature \cite{klartag2016discrete, munch2023non}, Cayley graphs over Abelian groups \cite{klartag2016discrete, oveis2021ARV}, or more generally Cayley graphs over nilpotent  
groups \cite{breuillard2016nilprogressions}, or graphs satisfying strong bounds on their volume growth \cite{benson2021volume}.
All these graphs satisfy $\phi(G) \geq \Omega(\sqrt{\lam_2 / d})$ where $d$ is the degree of the graph.

For solving these instances, it has been observed experimentally that spectral algorithms are less successful than flow-based algorithms
based on the Leighton--Rao LP and its successors \cite{leighton1999multicommodity, fountoulakis2023flow}.
For example, the computations of \cite[Table 2]{lang2009empirical} demonstrate that: for low-dimensional mesh graphs, Leighton--Rao succeeds while simple spectral thresholding algorithms are suboptimal, whereas the opposite is true for an ``expander with a planted cut'' similar to the graph in \cref{fig:three-communities}.

Note that the Buser Inequality $\phi(G) \geq \Omega(\sqrt{\lam_2/d})$ implies that Fiedler's Algorithm achieves $O(\sqrt{d})$-approximation to \sparsestcut
since the conductance of the Fiedler cut is $O(\sqrt{\lam_2})$\,.
Thus, Fiedler's Algorithm can still achieve constant approximation ratio
for constant $d$, although the approximation ratio decays with $d$.
The approximation ratio of Fiedler's Algorithm grows to $\Theta(\sqrt{\log n})$ for the hypercube graph (when it is perturbed slightly so that the $\{\pm 1\}$-indicator of the Majority function becomes the unique second eigenvector).

\subsection{Results}\label{sec:results}

In this work:
\begin{enumerate}

\item
We give a simple algorithm for the \cutimprovement problem using a modification of the Leighton--Rao linear program \cite{leighton1999multicommodity}.
This extends into a simple eigenspace enumeration algorithm for \sparsestcut in the same way as \cite[Section 4]{guruswami2013approximating}.

\item
We define a new complexity measure, the {\em solution dimension} $\CD$,
which determines the runtime of our eigenspace enumeration algorithm for \sparsestcut.
Comparing to \cref{thm:threshold-rank}, the solution dimension is no larger than the $O(\phi(G))$-threshold-rank, but surprisingly, it can be significantly smaller.
We demonstrate that this is the case for Cayley graphs over Abelian groups, which are a specific family of graphs satisfying the previously mentioned Buser Inequality. This gives us a PTAS for \sparsestcut\ in Abelian Cayley graphs of degree $d=o(\log\log n/\log\log\log n)$ and a subexponential time $(1+\eps)$-approximation for $d = o(\log n/\log \log n)$.

\item To prove our upper bound on the solution dimension,
we obtain several new facts about Abelian Cayley graphs.
We bound the number of sparse cuts in
an Abelian Cayley graph,
showing that all $O(1)$-approximate sparse cuts are $1-\eps$ contained in the span of the first $d^{O(d)}$ eigenvectors.
We also establish a new and tight bound of $2^{O(d)}$ on the multiplicity of $\lambda_2$\, improving a previous bound of $2^{O(d^2)}$ proven by Lee and Makaryachev~\cite{lee2008eigenvalue}\,.
In contrast to their proof, which is based on Kleiner’s proof of Gromov’s theorem \cite{kleiner2010new},
we directly show that the multiplicity of the $\lam_2$ eigenvalue in any regular graph can be robustly upper bounded by a notion of volume growth of the graph.

\end{enumerate}
We elaborate on these contributions below.

\smallskip

\salil{added subsubsection headings below}
\subsubsection{Cut Improvement}
Our first result is a new algorithm for the \cutimprovement problem.

\begin{theorem}
\torestate{
    \label{thm:arv-advice}
    Let $G$ be a regular $n$-vertex graph, let $0<\eps< 1/3$ and let $Q^*\subseteq V(G)$\,, $\vol(Q^*) \leq \vol(G)/2$. 
    There is a polynomial time algorithm that, given $G$ and $Q \subseteq [n]$ such that $\vol(Q \triangle Q^*) \leq \eps^3 \vol(Q^*)\,,$
    returns $\hat{Q}\subseteq [n]$ with $\psi(\hat{Q})\leq (1 + O(\eps))\cdot \psi(Q^*)\,.$
}
\end{theorem}
Compared to \cref{thm:cut-improvement} of Andersen and Lang~\cite{andersen2008algorithm}, this theorem offers a $(1+\eps)$-approximation to the density $\psi(Q^*)$ rather than to the conductance $\phi(Q^*)$, which will save us a factor of 2 in moving between the two.  
The algorithm of Andersen and Lang is based on iterative max flow computations on an augmented graph. In comparison, our algorithm seems conceptually simpler in that it uses the Leighton--Rao linear program with one extra constraint, followed by the simplest ball rounding scheme.
While other algorithms have been proposed for the \cutimprovement problem, such as a spectral algorithm by Mahoney et al \cite{mahoney2009spectral},
we are unaware of other algorithms besides Theorems \ref{thm:cut-improvement} and \ref{thm:arv-advice} that achieve $1+\eps$ approximation to the nearby sparse cut.

\subsubsection{Solution Dimension}
Our second contribution is a new characterization of graphs for which a combination of eigenspace enumeration  with an algorithm for \cutimprovement is an effective approach to \sparsestcut.
We name the parameter controlling the runtime of our spectral enumeration algorithm
the \emph{solution dimension} of the graph.
For a subspace $S\subseteq \R^n\,,$ let $\matPi_S\in \mathbb{R}^{n\times n}$ be the projection onto $S\,,$ let  $C_{\eps}(S) := \{x \in \R^n \, : \, \norm{x} =1, \norm{\matPi_{S} {x}}^2 \geq 1-\eps\}$ be the set of unit vectors near $S\,.$

\begin{definition}[Solution Dimension]\label{def:cut-dimension}
    Let $0\le \eps \le 1\,, c\geq 1\,,$ let $G$ be an
    $n$-vertex graph.
    Let $\matD \in \R^{n \times n}$ be the diagonal matrix whose $i$th entry is $\deg_G(i)$\,.
    Let $\lam_1 \leq \cdots \leq \lam_n$ be the sorted eigenvalues of the normalized Laplacian of $G$ and let $v_1,\dots, v_n \in \R^n$ be the associated eigenvectors.
    For $Q\subseteq [n]$, let $\bar{\mathbf{1}}_Q \in \R^n$ be the projection of $\mathbf{1}_Q$ orthogonal to the trivial eigenvector $v_1$.
    
    The $(\eps,c)$\emph{-solution-dimension} of $G$, denoted by $\CD_{\eps,c}(G)\,,$ is the smallest $k\in [n]$ such that
    there exists $Q \subseteq [n]$ with:
    \begin{enumerate}[(i)]
        \item $\psi_G(Q) \leq c \cdot \psi(G)\,.$
        \item $\matD^{1/2}\bar{\mathbf{1}}_Q / \norm{\matD^{1/2}\bar{\mathbf{1}}_Q} \in C_\eps(\textnormal{span}(v_2, \ldots,v_{k+1}) )$ .\footnote{Because of the centering, it does not matter whether the indicator of $Q$ is represented as a 0/1 vector or a $\pm 1$ vector, and $Q$ and $\bar Q$ are treated equivalently. For regular graphs, $\matD$ can be ignored.}
    \end{enumerate}
    When $c=1$ we simply write $\CD_{\eps}(G)\,.$ 
\end{definition}
\salil{moved the following paragraph from later up to here; it seems like it should precede the later discussion about the number of solutions. please double-check the quantitative bounds in the first line; I moved the 2 into the subscript}
It can be shown that $\CD_{\eps}(G) \leq \mul_{\psi(G)/\eps}(G) \leq \mul_{2\phi(G)/\eps}(G)$; indeed {\em every} sparse cut $Q$ can be approximated in the span of eigenvectors of eigenvalue at most $\psi(G)/\eps$.
(See \cref{sec:threshold-rank-vs-cut-dimensions} for a proof.)
Surprisingly, the $\phi(G)$-threshold-rank does \emph{not} always
match the solution dimension, allowing for potentially large speedups over \cref{thm:threshold-rank}.
The cycle graph is perhaps the most basic example with solution dimension $O(1)$ whereas $\mul_{\phi(G)} = \Theta(\sqrt{n})$, as we compute in \cref{app:cycle-graph}.
Put another way, it suffices for eigenspace enumeration to find a $(1-\eps)$ fraction of a sparse cut, instead of the entire cut,
which can dramatically reduce the runtime of the enumeration.

\salil{next sentence is a rephrasing of what we had before}
Furthermore, while $\mul_{\psi(G)/\eps}(G)$ being small implies that there are few ``distinct'' sparsest cuts (since all of them
are $\eps$-close to a low-dimensional subspace), small solution dimension only requires
that \emph{some} approximate sparsest cut is close to the low eigenspace.  

From this perspective, understanding whether graphs with small solution dimension yet many distinct sparsest cuts exist---and what their structure might be---appears to be a potentially valuable direction for developing new algorithms for \sparsestcut.

Combining \cref{def:cut-dimension} with \cref{thm:arv-advice} we deduce:

\begin{theorem}\torestate{
\label{thm:main-approximate-cut-dimension}
For all $0 \le \eps<1/3\,, c\geq 1\,,$ \sparsestcut admits a $c\cdot (1+O(\eps))$-approximation algorithm in time $n^{O(1)}\cdot \exp\set{\CD_{\eps^3,c}(G)\cdot O(\log 1/\eps)}\,.$
}
\end{theorem}

\subsubsection{Abelian Cayley Graphs}
Our next contribution is to show that using the solution dimension framework gives a dramatic speedup generically for low-degree Cayley graphs over Abelian groups.

\begin{definition}[Cayley Graph]
    \torestate{\label{def:abelian-cayley-graph}
    Let $\Gamma$ be a group and let $S$ be a multiset (called the set of \emph{generators}) from $\Gamma$ such that the multiplicity of $x \in S$ and $-x \in S$ is the same for all $x \in \Gam$. The \emph{Cayley graph} of $\Gamma$ generated by $S$, denoted $\cay(\Gamma, S)$, is the graph with vertex set $\Gamma$ and edge set $\{(v, v\cdot s): v \in \Gamma, s\in S\}$.
    }
\end{definition}

Even when restricted to constant degree,
Cayley graphs provide a simple way to construct graphs with interesting algebraic and geometric properties.
For example, the first expander construction by Margulis
is a 12-regular Cayley graph over the group $SL_3(\Z_p)$ \cite{margulis1973explicit}.

Constant-degree Cayley graphs over Abelian groups, and more generally nilpotent groups, are well-known in computer science and geometric group theory as \emph{non-expanding} graphs which have slow mixing properties.
Degree-$d$ Abelian Cayley graphs satisfy $\phi(G) \leq O(n^{-2/d})$ \cite{klawe1981non, friedman2006spectral},
the Cheeger upper bound is nearly tight due to the Buser inequality $\phi(G) \geq \Omega(\sqrt{\lam_2/d})$ \cite{klartag2016discrete, oveis2021ARV}, the volume of a radius-$r$ ball is at most $O(r^{d/2})$\,,
and they are flat or positively curved using any of several definitions of discrete curvature \cite{curvature_ricci_flatness, klartag2016discrete}.

\medskip

Trevisan posed the question of whether \sparsestcut admits an $O(1)$-approximation in polynomial time on Abelian Cayley graphs \cite{oveis2021ARV}.
We prove that higher-order spectral algorithms can
efficiently solve \sparsestcut on low-degree Abelian Cayley graphs using the solution dimension framework.

\begin{theorem}
    \torestate{\label{cor:apx-sparsest-cut-abelian}
        Let $G = \Cay(\Gam, S)$ be a Cayley graph over an Abelian group $\Gam, |\Gam| = n$ with generating set $S \subseteq \Gam, |S| = d$. There is an algorithm that finds a set $Q\subseteq [n], |Q| \leq n/2$ satisfying $\psi_G(Q)\leq (1 + \eps)\cdot \psi(G)$ in time $n^{O(1)}\cdot \exp\{(d/\eps)^{O(d)}\}\,.$
    }
\end{theorem} 
\salil{added next two paragraphs}
Recall that, because of the aforementioned Buser Inequality, Fiedler's algorithm gives an $O(\sqrt{d})$ approximation algorithm on degree $d$ Abelian Cayley graphs.  In contrast, \cref{cor:apx-sparsest-cut-abelian} gives a $(1+\eps)$ approximation ratio.  The price is that runtime is doubly exponential in the degree $d$.  Indeed, we obtain a PTAS for
degree $d=o(\log\log n/\log\log\log n)$ and a subexponential-time approximation scheme for $d = o(\log n/\log \log n)$.

Every finite abelian group $\Gam$ is isomorphic to a product of cyclic groups $\Gam\cong \Z_{n_1}\times\cdots\times \Z_{n_k}$.  The minimum number $k$ of cyclic factors in such a decomposition is also the minimum size of a generating set $S$ for $\Gam$.  Thus our algorithm is useful for groups $\Gamma$ that can be decomposed into $o(\log n/\log \log n)$ cyclic factors (such as cyclic groups). It is not applicable to $\Z_2^k$, since that requires $d\geq k=\log n$ generators.

To prove \cref{cor:apx-sparsest-cut-abelian} from \cref{thm:main-approximate-cut-dimension}, we bound the solution dimension of Abelian Cayley graphs by analyzing both their sparse cuts
and their low eigenvectors.
First we prove a generic bound on the multiplicity of eigenvalues near $\lam_2$ in terms of a notion of volume growth.

\begin{theorem}\torestate{\label{thm:dimension-low}
     Let $G$ be a regular graph. For every $\lam_2\leq \tau\leq \tfrac{3}{2}$,
    \begin{align*}
        \mul_\tau(G)\leq O\left(\frac{\tau}{\lambda_2}\right)^{\log \gam^{\CP}_G}
    \end{align*}
    where $\gam_G^{\CP} = \max_{t \geq 0} \frac{\CP_t}{\CP_{2t}}$ and $\CP_t$ is the return probability of a $2t$-step lazy random walk (\cref{def:collision-probability}).
    }
\end{theorem}

The quantity $\gam_G^{\CP}$ is bounded by $2^{O(d)}$ for degree-$d$ Abelian Cayley graphs, as we prove, which yields a bound of $2^{O(d)}$ on the multiplicity of the $\lam_2$ eigenvalue itself (by plugging in $\tau = \lam_2$). This
improves a previous bound of $2^{O(d^2)}$ proven by Lee and Makaryachev~\cite{lee2008eigenvalue}. \salil{added the which clause, seems important to repeat here}
Moreover, our bound is robust, in that we also bound the number of eigenvalues that are at most $\tau=c\cdot\lam_2$ by $O(c)^{O(d)}$.

The quantity $\log \gam_G^{\CP}$ appearing in \cref{thm:dimension-low} is essentially the rate of the \emph{cogrowth sequence}
from the geometric group theory literature \cite{pak2022algebraic}.
Obtaining bounds on eigenvalue multiplicity is also of independent interest in spectral graph theory, due in part to
a recent work by Jiang \etal which resolved a longstanding open question in geometry
about the maximum number of equiangular lines in $\R^d$\,, using a key insight
that the maximum multiplicity of $\lam_2$ in a graph with maximum degree $d$ is $o_d(n)$ \cite{jiang2021equiangular, mckenzie2021support, jiang2023spherical}.

For Abelian Cayley graphs with degree $d$ which is equal to a multiple of $\log n$, our bound meets the trivial upper bound $\mul_{\lam_2} \leq n$.
We observe that this upper bound is tight up to the constant factor in the exponent, since there exist Cayley graphs over $\Z^k_2$ coming from linear error-correcting codes which have $d = O(\log n)$ and eigenvalue multiplicity $n^{\Omega(1)}$. We prove this in \cref{app:codes}. 

\medskip

The second key ingredient behind \cref{cor:apx-sparsest-cut-abelian} is a bound on sparse cuts in low-degree Abelian Cayley graphs, established by the following theorem.

\begin{theorem}\label{thm:bounded-cut-dimension}
    Let $G = \Cay(\Gam, S)$ be a Cayley graph over an Abelian group $\Gam, |\Gam| = n$ with generating set $S \subseteq \Gam, |S| = d$.
    Then $\CD_{\eps}(G) \leq \mul_{\tau} (G)$ for $\tau = O(d\cdot \phi^2(G) /\eps^2)\,.$
\end{theorem}

Combining \cref{thm:dimension-low} and \cref{thm:bounded-cut-dimension} with the Cheeger inequality $\phi^2(G) \leq O(\lam_2)$ returns the bound $\CD_\eps(G) \leq (d/\varepsilon^2)^{O(d)}\,$ on the solution dimension.
This implies the final result in \cref{cor:apx-sparsest-cut-abelian}.
When $d \cdot \phi^2(G) \ll \phi(G)$ this theorem can lead to a large gap between $\CD_\eps(G)$ and $\mul_{\phi(G)}$ and consequently to the large speedup we observe over \cref{thm:threshold-rank}.

In fact, our proof shows a stronger result that \emph{all} of the sparsest cuts are $(1-\eps)$-contained in the span of the first $d^{O(d)}$ eigenvectors.
That is, all of the sparsest cuts are $(1-\eps)$-close to a hyperplane cut
in the $d^{O(d)}$ dimensional spectral embedding of $G$.
This upper bound on the number of sparse cuts is likely of independent interest.
In particular, a line of recent work \cite{bafna2021playing, bafna2023solving} shows how to solve \uniquegames instances over graphs with ``certifiable'' upper bounds on the number of solutions.
We consider it likely that our results can therefore be extended to solve \uniquegames instances in polynomial time on constant-degree Abelian Cayley graphs.\footnote{Bafna and Minzer \cite{bafna2023solving} informally define a ``globally hypercontractive graph'' to be one with a succinct and algorithmic characterization of its small non-expanding sets. Our result informally shows that low-degree Abelian Cayley graphs have this property, although the formal definitions do not exactly meet.}

\medskip
Regarding Cayley graphs of larger degree, $d = \Omega(\log n)$ appears to be a natural threshold for our analysis. Indeed, all Abelian Cayley graphs with $o(\log n)$ generators have $o(1)$ expansion by the bound $\phi(G) \leq O(n^{-2/d})$,
whereas a random Abelian Cayley graph with $2 \log n$ generators
will be an expander with high probability \cite{alon1994random}. 

Notably, Cayley graphs over $\Z_p^k$ with $p = O(1)\,,$ which may have small solution dimension, fall outside the range of structural results above, since they necessarily have $d \geq \Omega(\log n)$ in order to be connected.
Nonetheless, Cayley graphs over $\Z_p^k$ in fact admit an $O(p)$-approximation algorithm to \sparsestcut. This follows from \cite{improved_cheeger}, but we give a simpler algorithm and proof in \cref{prop:sparse_cut_vector_space}.

\medskip
In light of the results above, we conjecture that the class of Abelian Cayley graphs
does not contain hard examples for \sparsestcut:

\begin{conjecture}
    Let $G =\cay(\Gam, S)$ be a Cayley graph over an Abelian group $\Gam$ of size $n$. There is an algorithm that finds a set $Q\subseteq [n], |Q| \leq n/2$ satisfying $\phi_G(Q)\leq O(\phi(G))$ in time $n^{O(1)}$.
\end{conjecture}

A subexponential time algorithm would already be an interesting result,
considering that there is a subexponential time algorithm for \uniquegames \cite{arora2015subexponential}, but efforts to lift this algorithm back to \sparsestcut have not yet succeeded.

\subsection*{Organization}
The rest of the paper is organized as follows. In \cref{sec:techniques} we present the high level ideas behind our results. In \cref{sec:preliminaries} we introduce preliminary notions and definitions.  \cref{sec:arv-advice} contains the proof of \cref{thm:arv-advice}. \cref{sec:low-eigenspace} contains the proof of \cref{thm:dimension-low},  the matching lower bound is proved in \cref{app:codes}. \cref{sec:sparsest-cut-live-low-eigenspace} contains the proof of \cref{thm:bounded-cut-dimension}. 
In \cref{app:sparsest-cut-vector-space} we consider the restricted case of Cayley graphs over $\Z^k_p$.
The appendices contain results which flesh out the exposition.

\section{Techniques}\label{sec:techniques}
We present here the main ideas behind our results. 

\subsubsection*{Cut improvement (Theorem \ref{thm:arv-advice}) and eigenspace enumeration (Theorem~\ref{thm:main-approximate-cut-dimension})}

The starting point for \cref{thm:arv-advice} is the natural LP relaxation for \sparsestcut studied in \cite{leighton1999multicommodity} (see \cref{sec:arv-advice} for the precise definition of the program). This LP relaxation assigns lengths to edges such that the average distance between all pairs of vertices is fixed, while the average edge length is minimized.
The pairwise distances are constrained to form an $n$-point semi-metric $d(\cdot,\cdot)\,:V(G)\times V(G)\to\R\,$.   
The Leighton-Rao rounding procedure then computes a randomized line embedding $f:V(G)\to\R$ that is $1$-Lipschitz and has small average distortion with respect to this metric:
\begin{align*}
     \Omega(1/\log n)\sum_{\ij \in V(G)} d(i,j) \leq \sum_{\ij \in V(G)} \abs{f(i)-f(j)}\leq \sum_{\ij \in V(G)} d(i,j)\,.
\end{align*}
Finally, the line embedding can be easily converted into a cut, yielding a cut with sparsity at most $O(\log n)\cdot\psi(G)\,.$ 
The improvement of \cite{arora2009expander} comes from the fact that their program returns a metric of negative type and, for such metrics, one can efficiently construct a $1$-Lipschitz line embedding with average distortion at most $O(\sqrt{\log n})\,.$
Hence, natural paths to improve over these results may require better constructions of the metric and the line embedding,
which appears challenging and may be computationally hard in general \cite{raghavendra2012reductions}.

Our key insight is that, if we enforce solutions to be close to a given partition $(Q, V(G) \setminus Q)$ in the sense that:
\begin{align}\label{eq:constraint}
    \sum_{\ij \in Q}d(i,j)+\sum_{\ij \notin Q}d(i,j) \leq \eps |Q| \cdot |V(G) \setminus Q|\,,
\end{align}
then any feasible metric must have two ``well-separated'' sets of size at least $(1 - \eps^{\Omega(1)})|Q|$ and $(1 - \eps^{\Omega(1)})|V(G) \setminus Q|$ which are at distance at least $1 - \eps^{\Omega(1)}\,.$
Any such feasible metric admits an efficiently computable line embedding with distortion at most $1+\eps^{\Omega(1)}$ which can then be easily rounded into an integral solution.
In this approach, all feasible solutions must have large overlap with $(Q, V(G)\setminus Q)$ and hence the overall quality of the output depends on the starting point $Q$. 

Finally, to obtain \cref{thm:main-approximate-cut-dimension}, we simply initialize the program at a good $Q$, by finding in time $\exp\set{\CD_{\eps}(G)}$ a partition $(Q, V(G)\setminus Q)$ that differs from a sparsest cut in at most an $\eps$-fraction of the vertices. For this choice of the constraint \cref{eq:constraint} we can simultaneously ensure that the returned metric is easily roundable \textit{and} that there is a nearby optimal integral solution which we $(1 + \eps^{\Omega(1)})$-approximate. 

\subsubsection*{Eigenvalue multiplicity (Theorem~\ref{thm:dimension-low}) and slow decay of collision probability}

The common approach to bounding the eigenvalue multiplicity of graph Laplacians boils down to relating the local volume growth of induced subgraphs with the spectrum of the whole graph \cite{jiang2021equiangular, lee2008eigenvalue, mckenzie2021support}. We remark that weaker bounds are also immediate consequences of higher order Cheeger inequalities \cite{louis2012many, lee2014multiway}.

Limiting our discussion to $\mul_{\lambda_2}(G)\,,$ in the context of  Abelian Cayley graphs the most relevant work is by Lee and Makarychev \cite{lee2008eigenvalue}. Their notion of volume growth is the \textit{doubling constant}: $\gamma_G:=\max_{t\geq 0}|B(2t)|/|B(t)|$, where $B(t)$ is the ball of radius $t$ about the identity element of $\Gamma$. By vertex transitivity the choice of the center of the ball is inconsequential.
The importance of the doubling constant stems from the observation that it can be used to bound the packing number $\cN(t)$ of the graph $G\,.$
Building on a series of existing results \cite{colding1997harmonic, kleiner2010new, gupta2003bounded}, Lee and Makarychev use this observation to establish 
\begin{align*}
    \mul_{\lambda_2}(G)\leq \cN\left(1/(\gamma_G^{O(1)}\cdot\sqrt{\lambda_2})\right)\leq \gamma_G^{O(\log\gamma_G)}\,,
\end{align*}
where the second step leveraged the inequality $\lambda_2\leq O(\log(\gamma_G)/\diam(G))^2\,.$ 
Because Abelian Cayley graphs are known to have polynomial growth $\gamma_G\leq 2^{O(d)}$ \cite{diaconis1994moderate}, this allows them to conclude $\mul_{\lambda_2}(G)\leq 2^{O(d^2)}\,.$

\bigskip

We take a much more direct (and arguably significantly simpler) approach which directly relates spectral and combinatorial quantities.
The starting point is the notion of $t$-step collision probability $\CP_t\,.$
The $t$-step collision probability equals the probability that two independent walks of length $t$ have the same endpoint when starting from the same vertex, sampled from the stationary distribution of $G$.\footnote{A technical detail is that our analysis uses \textit{lazy} random walks.} Thus $1/\CP_t$ serves as a probabilistic analog of the size of the ball $B(t)$ (a similar notion was also used in \cite{mckenzie2021support}).
This motivates us to consider the following ``smooth'' version of the doubling constant: $$\gamma^{CP}_G =\max_{t\geq 0}\frac{\CP_t}{\CP_{2t}}.$$ 

An important fact is that the collision probability  equals the average of the eigenvalues of the adjacency matrix of $G^{2t}\,,$ the ${2t}$-th power of our starting graph. That is, we have for $t\geq 0\,,$
\begin{align}\label{eq:techniques-collision-probability}
    \CP_t=\tfrac{1}{n}\sum_{i=1}^n(1-\lambda_i)^{2t}\,.
\end{align}
This is the $2t$-th norm of the eigenvalues, also called the $2t$-th Schatten norm of the adjacency matrix of $G\,.$
The crucial consequence of the expressiveness of \cref{eq:techniques-collision-probability} is that the collision probability ratio is then entirely captured by the spectrum of $G\,:$ \begin{align}\label{eq:techniques-collision-probability-ratio}
    \frac{\CP_t}{\CP_{2t}}=\frac{\sum_{i=1}^n(1-\lambda_i)^{2t}}{\sum_{i=1}^n(1-\lambda_i)^{4t}}\,.
\end{align}
And now the key insight is that for an appropriate choice of $t=\Theta(\log\mul_{\lambda_2}(G)/\lambda_2)\,,$ \cref{eq:techniques-collision-probability-ratio} relates the collision probability ratio, the multiplicity of $\lambda_2$ and the eigenvalues of $G\,:$
\begin{align*}
    \mul_{\lambda_2}(G)^{\Omega(1)}\leq \frac{\CP_t}{\CP_{2t}}\leq \gamma^{CP}_G\,.
\end{align*}
\tom{I think the following should be ok}
The observation that, by taking  longer random walks, we can learn critical information about the spectrum and avoid constructing a large packing, is the fundamental improvement over \cite{lee2008eigenvalue}.
Our inequality is valid for all regular graphs and hence immediately implies a bound on the eigenvalue multiplicities for all graphs that satisfy a bound on the collision probability ratio.

\smallskip

Our second key contribution is an upper bound on $\gamma^{CP}_G$ for all Abelian Cayley graphs. To prove this statement, our challenge is to relate walks (not balls) of length $t$ to those of length $2t$ in an Abelian Cayley graph.
First, we can see that in an Abelian Cayley graph, the endpoint of a walk is completely determined by how many times each generator is chosen, and not on their order.
Therefore, we may replace the random walk by a random draw from a multinomial distribution, using $t$ samples of $d$ items each occurring with probability $1/d$.

From here we can directly relate the multinomial density for $t$ to that of $2t$.

The binomial densities for $t$ and $2t$ are  approximately Gaussians $X$ and $Y,$ where $Y$ has twice variance of $X$,
leading to a direct comparison inequality on the probability density functions, $p_Y(x) \le 2 p_X(x)$ for all $x \in \R$.
For $d$ generators, this idea can be extended (with additional arguments) to obtain a comparison equality with factor $2^{O(d)}$ which then implies 
\begin{align}\label{eq:techniques-eig-multiplicity}
    \mul_{\lambda_2}(G)^{\Omega(1)}\leq \gam^{CP}_G \leq 2^{O(d)}\,.
\end{align}

The bound in \cref{eq:techniques-eig-multiplicity} is best possible up to the constant in the exponent: there exist Abelian Cayley graphs $\Cay(\Z^k_2, S)$ with $\mul_{\lambda_2}\geq 2^{\Omega(d)}$. The construction of such graphs comes from a characterization of eigenvalue multiplicity in Cayley graphs over $\Z^k_2$ in terms of binary linear codes. Every Cayley graph $\Cay(\Z^k_2, S)$ corresponds to a binary linear code of dimension $k$ and block length $|S| = d$ with the multiplicity of $\lambda_2$ corresponding to the number of code words of minimum weight. Thus, the problem of constructing Abelian Cayley graphs with large eigenvalue multiplicity is equivalent to constructing binary linear codes with many code words of minimum weight. Using algebraic geometry codes, Ashikhmin, Barg, Vl\u adu\k t \cite{ashikhmin2001linear} constructed a family of binary linear codes with linear rate and distance, and with an exponential number of minimum-weight codewords, which furnishes the desired construction.

\subsubsection*{Sparse cuts are approximately low-dimensional (Theorem~\ref{thm:bounded-cut-dimension})}

To show that a partition $(Q, V(G)\setminus Q)$ in a graph $G$ approximately lives in the span of the first few eigenvectors, we observe that it suffices to relate its conductance in $G$, with its conductance in $G^{2t}$ in the sense:
\begin{align}\label{eq:techniques-expansion-growth}
    \phi_{G^{2t}}(Q) \leq \sqrt{t}\cdot \ell \cdot \phi_G(Q),
\end{align}
for some sufficiently large $t>0$ and small $\ell>0\,.$
Recall now that the expansion can be expressed spectrally as the Rayleigh quotient of the vector $\mathbf{1}_Q$. Let $\mathbf{1}_Q=\sum_{i\geq 1}q_i v_i$ be the representation of the indicator vector of the set $Q$ in the eigenbasis $v_1,\ldots, v_n$ of the graph. Then the above argument implies
\begin{align*}
    \tfrac{1}{|Q|}\sum_{i = 1}^n q_i^2\cdot (1-(1-\lambda_i)^{2t}) \leq \sqrt{t}\cdot \ell \cdot \phi_G(Q)\,.
\end{align*}
From this inequality, simple manipulations show that with the choice $t = \Theta(\eps^2\ell^{-1}\phi^{-2})$ all but $\eps$ of the spectral mass of $\bar{\mathbf{1}}_Q$ must be contained on eigenvalues up to $1/t$.

To establish a version of \cref{eq:techniques-expansion-growth} for Abelian Cayley graphs,
we leverage the analysis of Buser inequality \cite{klartag2016discrete, oveis2021ARV}.
The proof by Oveis Gharan and Trevisan \cite{oveis2021ARV} studies the probability of a length-$2t$ random walk $X_0,\ldots,X_{2t}$ crossing an arbitrary cut $Q$.
This measures the expansion of $Q$ in the graph $G^{2t}$ and their combinatorial analysis proves that \cref{eq:techniques-expansion-growth} holds for $\ell=\Theta(\sqrt{d})\,.$
As before, the crucial property of Abelian Cayley graphs being used is that it suffices to count how many times each generator is used by the walk.
In the proof of this inequality, it is also important that generators $g$ and $-g$ cancel out, which allows to bound the expected number of steps in the direction of generator $g$ by only $O(\sqrt{t/d})$.
Finally, because this argument does not rely on the specific cut $Q\,,$ we immediately obtain that for Abelian Cayley graphs all sparsest cut approximately live in the span of the first $O(\phi^2\cdot d/\eps^2)$ eigenvectors.

\section{Preliminaries}\label{sec:preliminaries}

We establish the notation used throughout the paper along with some preliminary notions.

The norm $\norm{\cdot}$ is the $\el_2$ norm.
For a subspace $S\subseteq \R^n\,,$ we denote by $\matPi_S$ the orthogonal projection onto $S.$ 
For a set $Q\subseteq [n]$ we denote by $\mathbf{1}_Q \in \{0,1\}^n$ its indicator vector.

In this paper, we study undirected graphs which may have self-loops or multiedges, which we refer to as graphs. We always use $n$ to denote the number of vertices in a graph $G$
and we assume $V(G) = [n]$.
We denote by $\matA(G)\in \R^{n\times n}$ the normalized adjacency matrix,
$$\matA(G)_\ij=\begin{cases}\tfrac{1}{\sqrt{\deg_G(i)\deg_G(j)}}&\text{ if $\ij \in E(G)$}\\0&\text{ otherwise\,,}\end{cases}$$
where $\deg_G(i)$ is the degree of $i$ in $G$ (counting a self-loop as degree 1).
When the graph is regular we use $d$ to denote its degree.
In some of these definitions we may omit $G$ when the context is clear.

Let $\matD(G)$ be the diagonal matrix with entries $\deg_G(i)$.
The normalized Laplacian of $G$ is defined as $\matL(G)=\Id -\matA(G)\,.$

The eigenvalues of a graph $G$ are the eigenvalues of its normalized Laplacian $\matL(G)$ and are denoted $0=\lambda_1(G)\leq \lambda_2(G)\leq\ldots\leq \lambda_n(G)$ to the associated eigenvectors $v_1(G),\ldots,v_n(G)$. We use $1=\alpha_1(G)\geq \ldots\geq \alpha_n(G)$ for the eigenvalues of the normalized adjacency matrix, in descending order. Note that for all $i\in[n]$, $\alpha_i(G) = 1-\lambda_i(G)$.
When $G$ is regular, $\matA(G)$ equals $\matW(G)$ which denotes the transition matrix of the simple random walk on $G$. When $G$ is not regular, $\matA(G)$ and $\matW(G)$ are similar (under conjugation by $\matD^{1/2}$) so they still have the same eigenvalues.

For a set $Q\subseteq V(G)$ we let $\partial Q = \{(i,j) \in E(G): i \in Q, j \notin Q\}\,$.
For $t\in \N$ we denote by $G^t$ the multi-graph obtained by taking an
edge for each length-$t$ walk in $G$. In other words, $\matA(G^t) = \matA(G)^t$.

\begin{fact}
    For all graphs $G$,
    $\matL(G)$ is a symmetric positive semidefinite (PSD) matrix with all eigenvalues
    in the range $[0,2]$.
\end{fact}

\begin{fact}\label{fact:matrix-power}
    Let $\mathbf{M}\in \R^{n\times n}$. If $\lambda$ is an eigenvalue of $\mathbf M$, then $\lambda^t$ is an eigenvalue of $\mathbf M^t\,.$ Furthermore, if $\mathbf M\sge 0$ then for any $i\in [n]\,,$ $\lambda_i(\mathbf M^t)=\lambda_i(\mathbf{M})^t\,.$
\end{fact}
An immediate corollary is the following relation between the spectrum of $G$ and its powers.

\begin{fact}\label{fact:spectrum-graph-power}
    Let $G$ be a graph and let $t\in \N\,.$ For any $i\in [n]$ it holds
    \begin{align*}
        \lambda_i(G^t) = 1 - \Paren{1-\lambda_i(G)}^t\,.
    \end{align*}
\end{fact}

We will make use of the following definition for the span of eigenvectors associated with small eigenvalues of the Laplacian.

\begin{definition}[$\low$ eigenspace]\label{def:low-eigenspace}
    For a graph $G$ and $0 \leq \tau \leq 2$ we define $\low_\tau(G)$ to be the span of the eigenvectors of $G$ associated with eigenvalues $\lambda\leq \tau\,.$ Notice that $\mul_\tau(G)=\dim(\low_\tau(G))\,.$
\end{definition}

The conductance and normalized density of a cut are defined for general graphs as follows.

\begin{definition}[Conductance and density]\label{def:conductance-prelim}
    Let $G$ be a graph and let $Q \subseteq V(G)$. The \emph{conductance} $\phi_G(Q)$ and the \emph{(normalized) density} $\psi_G(Q)$ are:
    \begin{align*}
        \phi_G(Q) := \frac{\card{\partial Q}}{\vol(Q)} \qquad \qquad \psi_G(Q) := \frac{|E(Q, \bar{Q})|/|E(G)|}{2\left(\frac{\vol(Q)}{\vol(G)}\right)\left(\frac{\vol(\bar Q)}{\vol(G)}\right)}
    \end{align*}
    where the \emph{volume} of $Q\subseteq V(G)$ is $\vol(Q) = \sum_{i\in Q}\deg_G(i)$\,.
    The conductance and density of $G$ are then $\phi(G):=\min_{Q\subseteq V: \vol(Q) \leq \vol(G)/2}\phi_G(Q)$ and $\psi(G) := \min_{Q \subseteq V} \psi_G(Q)$\,.
\end{definition}
The denominator of $\psi_G(Q)$ is the probability that two independent random vertices chosen according to the stationary distribution cross the $(Q,\bar Q)$ cut, and the numerator is the probability that a random edge in the graph crosses the $(Q,\bar Q)$ cut.
Although $\partial Q$ and $E(Q, \bar Q)$ are the same, we use
both notations to highlight that $\psi_G(Q)$ is symmetric with respect to complementing $Q$ whereas $\phi_G(Q)$ is not.

\begin{fact}\label{fact:sparsity-vs-conductance}
Let $G$ be a graph and $Q\subseteq[n]\,, \vol(Q)\leq \vol(G)/2\,$. Then,
\begin{align*}
    \phi_G(Q) \leq \psi_G(Q)\leq 2\phi_G(Q)\,.
\end{align*}
\end{fact}
\begin{proof}
    We have
    \[
        \psi_G(Q) = \frac{|E(Q, \bar Q)|/ |E(G)|}{2 \left(\frac{\vol(Q)}{\vol(G)}\right)\left(\frac{\vol(\bar Q)}{\vol(G)}\right)} = \frac{|E(Q, \bar Q)|}{ \vol(Q)\cdot \left(\frac{\vol(\bar Q)}{\vol(G)}\right)}\,.
    \]
    The condition $1/2 \leq \vol(\bar Q)/\vol(G) \leq 1$ implies that:
    \begin{align*}
        \frac{|E(Q,\bar Q)|}{\vol(Q)} &\leq \psi_G(Q) \leq 2\cdot \frac{|E(Q,\bar Q)|}{\vol(Q)}\\
        \iff \quad \phi_G(Q) &\leq \psi_G(Q) \leq 2 \phi_G(Q)\,.
    \end{align*}
\end{proof}

These combinatorial quantities can be interpreted spectrally as simple Rayleigh quotients. 
$\matL(G)$ always has a ``trivial'' eigenvalue $\lam_1 = 0$ to the eigenvector $v_1$ with entries $\sqrt{\deg_G(i)}$\,.
Let $\bar{x}$ denote the projection of $x \in \R^n$ orthogonal to $v_1$:
\[\bar{x} = x - \tfrac{\ip{x, v_1}}{\ip{v_1, v_1}}\cdot v_1\]
\begin{fact}\label{fact:rayleigh-quotient}
    Let $G$ be a graph. For all $Q \subseteq [n]$,
    \begin{align*}
    \phi_G(Q) &= \frac{\mathbf{1}_Q^T \matD^{1/2} \matL(G) \matD^{1/2}\mathbf{1}_Q}{\mathbf{1}_Q^T \matD \mathbf{1}_Q} \qquad  \left(\quad = \frac{\mathbf{1}_Q^T \matL(G) \mathbf{1}_Q}{\mathbf{1}_Q^T \mathbf{1}_Q} \qquad \text{if $G$ is regular}\right)\,,\\
    \psi_G(Q) &= \frac{\bar{\mathbf{1}}_Q^T \matD^{1/2} \matL(G) \matD^{1/2}\bar{\mathbf{1}}_Q}{\bar{\mathbf{1}}_Q^T \matD \bar{\mathbf{1}}_Q} \qquad \left(\quad = \frac{\bar{\mathbf{1}}_Q^T \matL(G) \bar{\mathbf{1}}_Q}{\bar{\mathbf{1}}_Q^T \bar{\mathbf{1}}_Q} \qquad \text{if $G$ is regular}\right)\,.
    \end{align*}
\end{fact}

Cheeger's inequality provides a quantitative relation between eigenvalues and conductance. \salil{it would be nice to update the paper to work with $\psi(G)$ rather than $\phi(G)$ in most places, but that can wait for a future version}

\begin{theorem}[Cheeger inequality]\label{thm:cheeger}
    Let $G$ be a graph. Then $\frac{\lambda_2}{2}\leq \phi(G)\leq \sqrt{2\lambda_2}\,.$
\end{theorem}

We will make use of Stirling's approximation:
\begin{fact}[Stirling's approximation~\cite{stirlingWikipedia}]\label{fact:stirling}
    $2\sqrt{t}\left(\frac{t}{e}\right)^t \leq t!\leq 2\sqrt{2t}\left(\frac{t}{e}\right)^t$ for all $t \in \N \setminus \{0\}$.
\end{fact}

We restate the definition of Abelian Cayley graphs.

\restatedef{def:abelian-cayley-graph}

Note that $G$ is $|S|$-regular and may have multiedges or self-loops.
Since $S$ is symmetric, we view $G$ as undirected.

\section{Cut Improvement algorithm}
\label{sec:arv-advice}

In this section we give an algorithm for the \cutimprovement problem.

\restatetheorem{thm:arv-advice}

Throughout the section we 
let $0\leq \eps \leq 1$, let $G$ be a graph with vertex set $V = [n]$, and let $Q^*, Q\subseteq V$ such that $\vol(Q^*) \leq \vol(G)/2$ and $\vol(Q \triangle Q^*) \leq \eps \vol(Q^*)$.
We prove the result for irregular graphs, but if one is just interested in regular graphs such as Cayley graphs, the notation can simplified using $\vol(Q) = d|Q|$\,.

Our algorithm is based on the Leighton--Rao linear program for the sparsest cut \cite{leighton1999multicommodity}, augmented with an ``advice set'' constraint. We start by describing the Leighton--Rao LP.  The variables of the LP are $d(u,v)$ for $u,v\in V$, which should be thought of as pair-wise distances of a semi-metric $d: V\times V \to \R$. The LP is then defined as:
\begin{align}\label{eq:lp}\notag
	&\text{minimize } \frac{\sum_{(u,v) \in E(G)} d(u,v)}{\sum_{u,v \in V}d(u,v) \deg(u)\deg(v)}\\
    \tag{\(\cP\)}\text{subject to } &\left \{
	\begin{aligned}
        &\sum_{u,v \in V} d(u,v)\deg(u)\deg(v) = \vol(Q^*) \vol(\bar{Q}^*)&\\
        &d(u,w) \leq d(u,v) + d(v,w); \; \forall u,v,w \in V && \text{(triangle inequality)}\\
        &0\leq d(u,v) \leq 1; \; \forall u,v \in V && \text{(non-negativity)}\\
        &d(u,v) = d(v,u); \; \forall u,v \in V && \text{(symmetric)}
	\end{aligned}
	\right \}=\colon
\end{align}
Note that we may assume to know $\vol(Q^*)$ as there are only $O(n^2)$ possible values.
Due to the normalization of the cut size in the first constraint, the denominator of the objective function is a constant and the objective becomes linear.
The remaining constraints enforce that $d$ is a semi-metric.

We then augment the Leighton--Rao LP with the following advice constraint:
\begin{equation}\label{eq:advice-constraint}\tag{\(\cP'(Q)\)}
    \sum_{u,v \in Q}d(u,v)\deg(u)\deg(v) + \sum_{u,v \in \bar{Q}} d(u,v)\deg(u)\deg(v) \leq 2\eps \vol(Q^*)\vol(\bar{Q}^*)
\end{equation}
The next lemma shows that the  above LP has a feasible solution of value $\frac{d}{n}\cdot \psi(Q^*)$ given the premise of \cref{thm:arv-advice}.

\begin{lemma}\label{lem:lp-value}
    \ref{eq:lp}$\cup$\ref{eq:advice-constraint} has a feasible solution with value $\frac{2|E(G)|}{\vol(G)^2} \cdot \psi(Q^*)$.
\end{lemma}
\begin{proof}
The candidate solution is $d(u,v) = |\mathbf{1}_{Q^*}(u) - \mathbf{1}_{Q^*}(v)|$ which has value $\frac{2|E(G)|}{\vol(G)^2} \cdot \psi(Q^*)$. It is easy to verify this semi-metric satisfies constraints in \ref{eq:lp}. The advice constraint is upper bounded by:
\begin{align*}
    & \qquad \sum_{u, v \in Q}d(u,v)\deg(u)\deg(v) + \sum_{u,v \in \bar Q}d(u,v)\deg(u)\deg(v)\\
    &= \sum_{\substack{u,v \in Q\\ \text{exactly one of $u,v \in Q^*$}}} \deg(u) \deg(v) + \sum_{\substack{u,v \in \bar Q\\ \text{exactly one of $u,v \in Q^*$}}} \deg(u) \deg(v)\\
    &\leq \vol(Q) \vol(Q \triangle Q^*) + \vol(\bar Q)\vol(Q \triangle Q^*)\\
    & \leq \eps \vol(G)\vol(Q^*) & (\vol(Q \triangle Q^*) \leq \eps \vol(Q^*))\\
    &\leq 2\eps \vol(Q^*)\vol(\bar Q^*)\,. & (\vol(Q^*) \leq \vol(G)/2)
\end{align*}
\end{proof}

We will next argue that any feasible solution to \ref{eq:lp}$\cup$\ref{eq:advice-constraint} can be rounded into a sparse cut.

To that end, we perform the following probabilistic analysis.  Consider the distribution $\mu$ defined by the following process: 
\begin{itemize}
    \item pick $u \in V$ with probability proportional to $\deg(u)$ and $t \overset{u.a.r.}{\sim} [0,1)$
    \item let $B$ be the ball $B := \{v \in V: d(u,v) \leq t\}\,.$
\end{itemize}

We will prove that:

\begin{align*}
    \frac{\E_{\mu} |E(B, \bar B)|}{ \E_{\mu} \vol(B)\vol( \bar B)} &\leq (1 + O(\eps^{1/3})) \cdot \frac{\sum_{(u,v) \in E(G)} d(u,v)}{\sum_{u,v \in V} d(u,v) \deg(u)\deg(v)}\,.
\end{align*}

By the fact that for non-negative random variables $X, Y$,  ${\Pr\left(\frac{X}{Y}\leq \frac{\E X}{\E Y}\right)>0}$ assuming $\Pr(Y > 0) > 0$, there exists $u,t$ such that the same inequality holds with the expectations $\E_\mu$ removed.
Without the expectations, the inequality becomes $\psi(B) \leq (1+O(\eps^{1/3}))\cdot \psi(Q^*)$ using \cref{lem:lp-value}.
The rounding algorithm will thus enumerate all possible vertices $u\in [n]$ and distances $t\in \{d(u,v)\}_{u,v\in V} $ and find the ball $B$ with minimum $\psi(B)$, which takes $O(n^3)$ iterations.
This will complete the proof of \cref{thm:arv-advice}.

\medskip

First we upper bound the numerator.
\begin{lemma}\label{lem:numerator}
    \[{\E}_{\mu} |E(B, \bar B)| \leq \sum_{(u,v)\in E(G)} d(u,v) \]

\end{lemma}
\begin{proof}
Let $u$ be the center of the ball $B$.
Observe that every edge $(v,w)$ is a cut edge with probability 
\[\Pr\left(d(u,v) \leq t \leq d(u,w)\right) = |d(u, v) - d(u, w)| \overset{*}{\leq} d(w, v) ~~~~~~ (*:\text{ triangle inequality})\,.\]
Therefore the claim follows.
\end{proof}

To finish the proof, we lower bound the denominator $\E_{\mu}\vol(B)\vol(\bar B)\,.$
\begin{lemma}
    $\E_\mu \vol(B)\vol(\bar B) \geq (1 - O(\eps^{1/3}))\vol(Q^*)\vol(\bar Q^*)$
\end{lemma}
\begin{proof}
Let $B_r(u) = \{v \in V: d(u,v) \leq r\}$ denote a ball around a vertex.
Define $V' \subseteq V$ to be the vertices such that the ball around each vertex contains many vertices, which we call ``heavy'' clusters,
\[
    V' = \{u \in Q : \vol(B_{\eps^{1/3}}(u) \cap Q) \geq (1 - \eps^{1/3})\vol(Q)\} \cup \{u \in \bar Q : \vol(B_{\eps^{1/3}}(u) \cap \bar Q) \geq (1 - \eps^{1/3})\vol(\bar Q)\}\,.
\]
The advice constraint \cref{eq:advice-constraint} implies that $V'$ is very large as shown in the next claim.
\begin{claim}\label{lem:heavy-clusters}
    $\vol(V') \geq (1 - O(\eps^{1/3}))\vol(G)$
\end{claim}
\begin{proof}[Proof of claim.]
We may lower bound \cref{eq:advice-constraint} by:
\begin{align*}
    &\quad 2\eps \vol(Q^*)\vol(\bar Q^*)\\
    &\geq \sum_{u,v \in Q} d(u,v)\deg(u)\deg(v) + \sum_{u,v \in \bar Q} d(u,v)\deg(u)\deg(v) & (\text{\cref{eq:advice-constraint}})\\
    &\geq   \vol(Q \setminus V') \cdot \eps^{1/3} \cdot \eps^{1/3}\vol(Q) + \vol(\bar Q \setminus V') \cdot \eps^{1/3} \cdot \eps^{1/3} \vol(\bar Q) & (\text{def. of $V'$})\\
    &= \eps^{2/3}\left(\vol(Q \setminus V') \vol(Q) + \vol(\bar Q \setminus V')\vol(\bar Q)\right)\\
    &\geq O(\eps^{2/3})\left(\vol(Q \setminus V') \vol(Q^*) + \vol(\bar Q \setminus V')\vol(\bar Q^*)\right)\,.
\end{align*}
The last line uses $\vol(Q) = (1 \pm O(\eps))\cdot \vol(Q^*)$ and $\vol(\bar Q) = (1 \pm O(\eps))\cdot \vol(\bar Q^*)$.
Therefore, we have $\vol(Q \setminus V') \leq O(\eps^{1/3})\vol(\bar Q^*)$ and $\vol(\bar Q \setminus V') \leq O(\eps^{1/3})\vol(Q^*)$. Adding these together,
\[
    \vol(V \setminus V') = \vol(Q \setminus V') + \vol(\bar Q \setminus V') \leq O(\eps^{1/3}) \vol(G)
\]
as claimed.
\end{proof}

Next define the subset $V'' \subseteq V$ by,
\[
    V'' = \{u \in Q : \vol(B_{1 - \eps^{1/3}}(u) \cap \bar Q) \leq \eps^{1/3}\vol(\bar Q)\} \cup \{u \in \bar Q : \vol(B_{1 - \eps^{1/3}}(u) \cap Q) \leq \eps^{1/3}\vol(Q)\}\,.
\]
The set $V''$ consists of vertices in $Q$ which are separated from almost all of $\bar Q$ and vice versa.
In the same way as \cref{lem:heavy-clusters}, we will bound $\vol(V'') \geq (1 - O(\eps^{1/3}))\vol(G)$.
\begin{claim}
    $\vol(V'') \geq (1 - O(\eps^{1/3})) \vol(G)$
\end{claim}
\begin{proof}[Proof of claim.]
Combining \cref{eq:advice-constraint} with the equation,
\begin{align*}
    \vol(Q^*)&\vol(\bar Q^*) = \sum_{u,v \in V}d(u,v)\deg(u)\deg(v)\\
    &= \sum_{u,v \in Q} d(u,v)\deg(u)\deg(v) + \sum_{u,v \in \bar Q}d(u,v)\deg(u)\deg(v) + \sum_{\substack{u \in Q\\v \in  \bar Q}}d(u,v)\deg(u)\deg(v)
\end{align*}
the advice constraint also implies that
\begin{align}
\sum_{\substack{u \in Q\\v \in \bar Q}}d(u,v)\deg(u)\deg(v) &\geq (1 - 2\eps)\vol(Q^*)\vol(\bar Q^*) \notag\\
\implies \sum_{\substack{u \in Q\\v \in \bar Q}}(1 - d(u,v))\deg(u)\deg(v) &\leq 2\eps\vol(Q^*)\vol(\bar Q^*) + \vol(Q) \vol(\bar Q) - \vol(Q^*)\vol(\bar Q^*)\,. \label{eq:lp2}
\end{align}

Using $\vol(Q) = (1 \pm O(\eps))\cdot \vol(Q^*)$ and $\vol(\bar Q) = (1 \pm O(\eps))\cdot \vol(\bar Q^*)$ we can bound the right-hand side by $O(\eps) \vol(Q^*)\vol(\bar Q^*)$.
    Now similarly to \cref{lem:heavy-clusters}, we may bound \cref{eq:lp2} by,
    \begin{align*}
        &\quad O(\eps) \vol(Q^*)\vol(\bar Q^*) \\
        &\geq \sum_{\substack{u \in Q\\v \in \bar Q}} (1 - d(u,v)) \deg(u) \deg(v) & (\text{\cref{eq:lp2}})\\
        & \geq \vol(Q \setminus V'') \cdot \eps^{1/3}\cdot \eps^{1/3}\vol(\bar Q) + \vol(\bar Q \setminus V'') \cdot \eps^{1/3} \cdot \eps^{1/3}\vol(Q)\\
        & = \eps^{2/3}(\vol(Q\setminus V'')\vol(\bar Q) + \vol(\bar Q \setminus V'')\vol(Q))\\  
        &\geq O(\epsilon^{2/3}) (\vol(Q\setminus V'')\vol(\bar Q^*) + \vol(\bar Q \setminus V'')\vol(Q^*))
    \end{align*}
    which implies $\vol(V \setminus V'') \leq O(\eps^{1/3})\vol(G)$ as claimed.
\end{proof}

Therefore, we have $\vol(V' \cap V'') \geq (1 - O(\eps^{1/3}))\vol(G)$.
Let us call the vertices in $V' \cap V''$ \emph{good}.

When the randomized ball rounding is centered at $u \in Q$\,, all the vertices in $B_{\eps^{1/3}}(u) \cap Q$ will be separated from those in $\bar Q \setminus B_{1 - \eps^{1/3}}(u)$ if additionally the distance threshold satisfies $t \in [\eps^{1/3}, 1 - \eps^{1/3}]$.
When $u$ is good, these sets are both very large,
in particular,
\begin{align*}
\vol(B_{\eps^{1/3}}(u) \cap Q) &\geq (1 - \eps^{1/3})\vol(Q)\\
\vol(\bar Q \setminus B_{1 - \eps^{1/3}}(u)) &\geq (1 - \eps^{1/3})\vol(\bar Q)\,.
\end{align*}
We conclude that,
\begin{align*}
    \E_{\mu'} \vol(B)\vol(\bar B) &\geq \Pr(u \text{ good}) \cdot \Pr\left(t \in [\eps^{1/3}, 1 - \eps^{1/3}]\right) \cdot (1 - O(\eps^{1/3}))\vol(Q) \vol(\bar Q)\\
    &\geq (1 - O(\eps^{1/3}))\cdot (1 - O(\eps^{1/3})) \cdot (1 - O(\eps^{1/3}))\vol(Q)\vol(\bar Q)\\
    & = (1 - O(\eps^{1/3}))\vol(Q)\vol(\bar Q)\\
    & = (1 - O(\eps^{1/3}))\vol(Q^*)\vol(\bar Q^*)\,.
\end{align*}
This completes the analysis of the denominator.
\end{proof}

\subsection{Eigenspace enumeration for Sparsest Cut}

Next, we show how to derive an eigenspace enumeration algorithm for \sparsestcut using the algorithm for \cutimprovement.
Recalling the definition of solution dimension $\CD_{\eps, c}$ from the introduction, we will prove:

\restatetheorem{thm:main-approximate-cut-dimension}

The algorithm is as follows. We assume to know $\CD_{\eps^3, c}(G)$\,.
\chris{I removed ``without loss of generality'' which was here before. It is not possible to enumerate all possible values of $\CD$ since we do not want to enumerate an unnecessarily large eigenspace.}

\begin{algorithmbox}\label{alg:main}
    \mbox{}\\
    \textbf{Input:} graph $G\,, 0\leq \eps\leq 1$\\
    \textbf{Output:} Set $\hat Q\subseteq [n]\,.$
    \begin{enumerate}[(1)]
        \item Compute the eigenvectors $v_1,\dots, v_n$ of the normalized Laplacian of $G$ and let $S := \linspan(v_1,\dots, v_{\CD_{\eps^3, c}})$\,.
        \item Enumerate a $\sqrt{\eps}$-net $N$ for the unit vectors in $S$ of size $O(1/\sqrt{\eps})^{\dim(S)}$.
        \item For each enumerated vector $w \in N$ and each possible threshold $\tau \in \R$, run the \cutimprovement algorithm (\cref{thm:arv-advice}) on the threshold cut $\{i \in [n]: w_i \geq \tau\sqrt{\deg_G(i)}\}$.
        \item Output the minimum sparsity cut seen among the candidates from the previous step.
    \end{enumerate}
\end{algorithmbox}

\begin{remark}[Running time]
    There are at most $O(1/\sqrt{\eps})^{\dim(S)}$ vectors in $N$,
    each of which has at most $n$ distinct threshold cuts,
    and the algorithm in \cref{thm:arv-advice} takes time $n^{O(1)}$,
    for an overall runtime of $n^{O(1)}\cdot O(1/\sqrt{\eps})^{\dim(S)}$.
\end{remark}

We use the following Lemma.

\begin{lemma}\label{lem:threshold}
    Let $0 \leq \eps < 1/8,\, v \in \R^n, \norm{v} = 1$ and $Q^* \subseteq [n]$ such that $\vol(Q^*) \leq \vol(G)/2$.
    If $\left\|\matD^{1/2}\bar{\mathbf{1}}_{Q^*} / \norm{\matD^{1/2}\bar{\mathbf{1}}_{Q^*}} - v\right\|^2 \leq \eps$ then $\norm{\matD^{1/2}\mathbf{1}_{Q} - \matD^{1/2}\mathbf{1}_{Q^*}}^2 = \vol(Q \triangle Q^*) \leq 8\eps \vol(Q^*)$ where $Q = \left\{i \in [n]: v_i \geq \sqrt{\tfrac{\deg(i)}{8\vol(Q^*)}}\right\}$.
\end{lemma}
\begin{proof}
    The entries of $\matD^{1/2}\bar{\mathbf{1}}_{Q^*}/ \norm{\matD^{1/2}\bar{\mathbf{1}}_{Q^*}}$ are either $\sqrt{\frac{\vol(\bar Q^*) \deg(i)}{\vol(Q^*) \vol(G)}}$ or $-\sqrt{\frac{\vol(Q^*)\deg(i)}{\vol(\bar Q^*)\vol(G)}}\,$.
    Since $\vol(Q^*) \leq \vol(G)/2$, these are at least $\sqrt{\tfrac{\deg(i)}{2\vol(Q^*)}}$ and at most 0 respectively.
    We have defined $Q$ by thresholding on the midpoint of these two values.
    
    Therefore, each disagreement between $Q$ and $Q^*$ leads to an entry of somewhat large magnitude in the vector $\matD^{1/2}\bar{\mathbf{1}}_{Q^*}/\norm{\matD^{1/2}\bar{\mathbf{1}}_{Q^*}} - v$\,. We have:
    \begin{align*}
        &\eps \geq \left\|\matD^{1/2}\bar{\mathbf{1}}_{Q^*} / \norm{\matD^{1/2}\bar{\mathbf{1}}_{Q^*}} - v\right\|^2 \geq \sum_{i \in Q \triangle Q^*} \frac{\deg(i)}{8 \vol(Q^*)} = \frac{\vol(Q \triangle Q^*)}{8\vol(Q^*)}\,.
    \end{align*}
    Rearranging proves the claim.
\end{proof}

Now we complete the proof.
For a subspace $S \subseteq \R^n$, recall that we define $C_\eps(S) := \{x \in \R^n \, : \, \norm{x} = 1, \norm{\matPi_S x}^2 \geq 1-\eps\}$ to be the unit vectors near $S$.

\begin{proof}[Proof of \cref{thm:main-approximate-cut-dimension}]
Let $\epsilon$ be small enough and let $Q^* \subseteq [n]$ be the cut with minimum sparsity such that $\matD^{1/2}\bar{\mathbf{1}}_{Q^*}/\norm{\matD^{1/2}\bar{\mathbf{1}}_{Q^*}} \in C_{\eps}(S)\,.$
We may assume $\vol(Q^*) \leq \vol(G)/2$ since $Q^*$ and $\bar Q^*$ are treated identically.
Let $v^* = \matPi_S\matD^{1/2}\bar{\mathbf{1}}_{Q^*} / \norm{\matPi_S\matD^{1/2}\bar{\mathbf{1}}_{Q^*}}$ be the projection of $\matD^{1/2}\bar{\mathbf{1}}_{Q^*}$ to $S$ rescaled into a unit vector.
The subspace enumeration gives a unit vector $v \in S$ such that
\[
    \norm{\matD^{1/2}\bar{\mathbf{1}}_{Q^*}/\norm{\matD^{1/2}\bar{\mathbf{1}}_{Q^*}} - v} \leq \underbrace{\left\|\matD^{1/2}\bar{\mathbf{1}}_{Q^*}/\norm{\matD^{1/2}\bar{\mathbf{1}}_{Q^*}} - v^*\right\|}_{\leq O(\sqrt{\eps})\text{ by defn. of }C_\eps(S)} + \underbrace{\left\|v^* - v\right\|}_{\leq \sqrt{\eps}\text{ by net}} \leq O(\sqrt{\eps})\,.
\]
By \cref{lem:threshold}, there is a threshold cut $Q = \{i \in [n]: v_i \geq \tau \sqrt{\deg(i)}\}$ such that $\vol(Q \triangle Q^*) \leq O(\eps) \cdot \vol(Q^*)$.
We may then use the set $Q$ as the advice in \cref{thm:arv-advice} to obtain the desired cut.
\end{proof}

\begin{remark}
Eigenspace enumeration can also be viewed as searching for a hyperplane cut in the spectral embedding.
That is, letting $\lambda_1\leq \ldots\leq\lambda_n$ be the sorted eigenvalues of the normalized Laplacian and $v_1,\ldots,v_n$ be the associated eigenvectors,
the $k$-dimensional spectral embedding maps vertex $i \in [n]$ to $((v_{1})_i, \dots, (v_{k})_i) \in \R^k$.
Then,
\[
\sum_{i = 1}^k c_i v_i \approx \pm \mathbf{1}_Q \qquad \Longleftrightarrow \qquad Q \approx \begin{array}{c}\text{hyperplane cut in $k$-dimensional spectral} \\ \text{embedding with normal vector }\vec{c} \in \R^k\,.\end{array}
\]

In this way, eigenspace enumeration is equivalent to brute-force searching over all possible hyperplane cuts in a certain low-dimensional embedding of the graph. It is interesting to compare this with the Goemans--Williamson hyperplane rounding technique which computes a random hyperplane cut in the SDP embedding of the graph \cite{goemans1995improved}.
\end{remark}

\section{The low eigenspace of Abelian Cayley graphs}\label{sec:low-eigenspace}

In this section, we study the low eigenspace of Abelian Cayley graphs and prove \cref{thm:dimension-low}. We do so by analyzing the collision probability
of a random walk in $G$. For a probability distribution $\pi$, the \textit{collision probability} of $\pi$ is defined by $\CP(\pi) = \|\pi\|^2_2=\bbP_{x,x'\sim \pi}(x=x')\,.$

\begin{definition}[$t$-step lazy collision probability]
\label{def:collision-probability}
    Let $G$ be a graph and $\pi$ be the stationary distribution for $\frac 1{2}\Id + \frac{1}{2}\matW$.
   The \emph{$t$-step lazy collision probability} is defined by $\CP_t = \E_{x\sim \pi}\Brac{\CP\left(\left(\frac{1}{2}\Id + \frac{1}{2}\matW\right)^t \mathbf{1}_{x}\right)}$\,.
\end{definition}

Recall that the stationary distribution is proportional to the degree,
and is the uniform distribution when the graph is regular. Note that for vertex-transitive graphs such as Abelian Cayley graphs, the collision probability satisfies $\CP_t = \CP((\frac 12 \Id + \frac 12\matW)^{t}\mathbf{1}_x)$ for all vertices $x$.
The next Lemma gives a spectral interpretation of $\CP_t$ as the power sums of the eigenvalues of the normalized adjacency matrix (a.k.a the moments of the empirical spectral distribution).

\begin{lemma}\label{lem:collision_sum_eigenvalues}
     Let $G$ be a regular graph. Then
     \[\CP_t = \frac{1}{n}\sum_{i=1}^n\left(1-\frac{\lambda_i}{2}\right)^{2t}.\]
     \begin{proof}
     Let $X_0, X_1, \dots, X_{2t}$ be a simple random walk initialized at $X_0$ drawn from the stationary distribution.
    Let $\tilde{X}_0, \dots, \tilde{X}_t$ be an independent simple random walk initialized at $X_0$.
    Since the simple random walk is a reversible Markov chain,
    $\CP_t = \Pr(X_t = \tilde{X}_t) = \Pr(X_{2t} = X_0)$.

    On the other hand, the diagonal elements of the transition matrix equal the returning probabilities of a random walk. Therefore
    \begin{align*}
        \CP_t = \E_{x\sim \pi}\Brac{\left(\tfrac 12 \Id + \tfrac 12\matW\right)^{2t}_{x,x}} = \frac 1n \Tr\left(\left(\tfrac 12 \Id + \tfrac 12 \matW\right)^{2t}\right) = \frac{1}{n}\sum_{i=1}^n\left(1-\frac{\lambda_i}{2}\right)^{2t},
    \end{align*}
    as desired. The second equality uses that $G$ is regular, and the last equality uses that $\matW$ and $\matA$ have the same spectrum.
 \end{proof}
\end{lemma}

To bound the multiplicity of eigenvalues close to $\lambda_2$, we analyze the ratio $\CP_t/\CP_{t(\kappa + 1)}$ where $t\in \N$ and $\kappa \geq 1$ are parameters.
Interestingly, the argument works for all regular graphs and not only for Abelian Cayley graphs.
We will derive \cref{thm:dimension-low} as a consequence of the next Lemma.

\begin{lemma}\label{lem:lower-bound-collision-probability}
    Let $G$ be a connected regular graph on $n$ vertices. Suppose $\lambda_2\leq \tau\leq \frac{3}{2}$ and let $\kappa = \lceil\tau/\lambda_2\rceil$. Then for $t =\lfloor\ln(\mul_\tau)/3\tau\rfloor$, we have
    \begin{align*}
        \frac{\CP_t}{\CP_{t(\kappa+1)}}\geq \mul_\tau^{1/3}/(2e^{3/2})\,.
    \end{align*}
\end{lemma}

\begin{proof}
    First notice that for any $i\geq 2\,,$ it holds $0\leq (1-\lambda_i/2)\leq (1-\lambda_2/2)$ since $\lambda_2\leq\lambda_i\leq 2\,.$ By applying  \cref{lem:collision_sum_eigenvalues}, we have 
    \[
    \frac{\CP_t}{\CP_{t(\kappa+1)}} = \frac{\sum_{i\in[n]}(1-\lambda_i/2)^{2t}}{\sum_{i\in[n]}(1-\lambda_i/2)^{2t(\kappa+1)}} = \frac{\sum_{i\in[n]}(1-\lambda_i/2)^{2t}}{1+\sum_{i\geq 2}(1-\lambda_i/2)^{2t(\kappa+1)}}.
    \]
    
    We show the following lower bound,
    \begin{equation}\label{eqn:cp_lwbd}
            \sum_{i\in[n]}(1-\lambda_i/2)^{2t}\geq \max\Paren{\mul_\tau \cdot e^{-2t\tau}, e^{t\tau}\sum_{i\geq 2}(1-\lambda_i/2)^{2t(\kappa+1)}}.
    \end{equation}
    First observe that if $\lambda_i\leq \tau$, then $(1-\lambda_i/2)^{2t}\geq (1-\tau/2)^{2t}$. Now by ignoring all $\lambda_i$ that are not in $\low_{\tau}$, we have
   \begin{align*}
       \sum_{i\in[n]}(1-\lambda_i/2)^{2t} &\geq \mul_\tau \cdot(1-\tau/2)^{2t}\\
       &\geq \mul_\tau \cdot e^{-2t\tau},
   \end{align*}
   where the last inequality uses the fact that $1-x/2\geq e^{-x}$ for $x\in\left[0, 3/2\right]$.
   Now, note that 
   \[(1-\lambda_i/2)^{2t} = (1-\lambda_i/2)^{-2t\kappa}\cdot (1-\lambda_i/2)^{2t(\kappa +1)}.\]
   This implies, $(1-\lambda_i/2)^{2t}\geq (1-\lambda_2/2)^{-2t\kappa}\cdot(1-\lambda_i/2)^{2t(\kappa+1)}$. In particular, we can use this to obtain 
   \begin{align*}
       \sum_{i\in[n]}(1-\lambda_i/2)^{2t} &\geq (1-\lambda_2/2)^{-2t\kappa}\cdot\sum_{i\geq 2}(1-\lambda_i/2)^{2t(\kappa+1)}\\
       &\geq e^{t\tau}\cdot\sum_{i\geq 2}(1-\lambda_i/2)^{2t(\kappa+1)}.
   \end{align*}
   The final inequality uses the fact that $1-x\leq e^{-x}$ and $\lambda_2\kappa\geq \tau$. This proves \eqref{eqn:cp_lwbd}. Observe, that for all non-negative numbers $a,b,c,d$ with $c,d>0$ we have $\max\{a,b\}/(c+d)\geq \frac{1}{2}\min\{a/c,b/d\}.$ This implies the following inequality
   \begin{equation}\label{eqn:cp_lwbd_2}
   \frac{\max\Paren{\mul_\tau \cdot e^{-2t\tau}, e^{t\tau}\cdot\sum_{i\geq 2}(1-\lambda_i/2)^{2t(\kappa+1)}}}{1+\sum_{i\geq 2}(1-\lambda_i/2)^{2t(\kappa +1)}} \geq \frac{1}{2}\min\Paren{\mul_\tau \cdot e^{-2t\tau}, e^{t\tau}}.
    \end{equation}

Choosing $t = \lfloor\ln(\mul_\tau)/3\tau\rfloor$ gives the desired inequality. The factor of $e^{3/2}$ in the denominator comes from the fact that $e^{\tau\lfloor \ln(\mul_\tau)/3\tau\rfloor}\geq e^{\tau\left(\frac{1}{3\tau}\ln(\mul_\tau)-1\right)}\geq \mul_\tau^{1/3}/e^{\tau}$ and the assumption $\tau\leq 3/2$.
\end{proof}

We deduce \cref{thm:dimension-low} as an immediate corollary.

\begin{proof}[Proof of \cref{thm:dimension-low}]
    Let $\gam_G^{\CP} = \max_{t \geq 0} \frac{\CP_t}{\CP_{2t}}$.
    Let $\kappa = \lceil \tau/\lambda_2 \rceil.$ Observe that by \cref{lem:upper-bound-collision-probability},
    \begin{align*}
      \frac{\mul_\tau^{1/3}}{(2e^{3/2})} \leq \frac{\CP_{t}}{\CP_{{t(\kappa + 1)}}} &\leq \frac{\CP_t}{\CP_{2t}}\cdot \frac{\CP_{2t}}{\CP_{4t}}\cdots\frac{\CP_{t2^{\lceil\log(\kappa+1)\rceil-1}}}{\CP_{t2^{\lceil\log(\kappa+1)\rceil}}}\\
      &\leq (\gam_G^{\CP})^{\lceil\log(\kappa +1)\rceil}\\
      &\leq (\gam_G^{\CP})^{\log(O(\tau/\lambda_2))},
    \end{align*}
    where the last inequality uses the fact that $\kappa \leq  2\tau/\lambda_2$ and $\log(2\tau/\lambda_2 +1) \leq \log(3\tau/\lambda_2)$. This implies,
    \begin{align*}   
    \mul_\tau&\leq (\gam_G^{\CP})^{\log(O(\tau/\lambda_2))+11} = O\left(\frac \tau{\lam_2}\right)^{\log \gam_G^{\CP}}.
    \end{align*}
\end{proof}

To specialize this to Abelian Cayley graphs,
we prove an upper bound on the ratio $\CP_t/\CP_{2t}$ with a function that depends on the degree $d$ of the graph but not on $t$ or $n$.
A version of this can be deduced for infinite, finitely-generated Abelian and nilpotent groups using tools related to Gromov's theorem, see \cite{pak2022algebraic} for references.
In comparison, our proof is self-contained, and we cover finite graphs.

\begin{lemma}\label{lem:upper-bound-collision-probability}
    Let $G = \cay(\Gamma, S)$ be a degree $d$ Abelian Cayley graph. Then, for every integer $t\geq 0\,,$ $\frac{\CP_t}{\CP_{2t}}\leq (2e)^{4d}\,.$
    \begin{proof}

    To simplify the analysis of the quantity $\CP_t$, we 
    \begin{enumerate}
        \item replace the lazy random walk with non-lazy random walk by introducing $d$ new copies of the identity element as generators, and 
        \item assume each generator occurs with even multiplicity. This can be done by making a copy of every generator. Note this does not change the random walk matrix and hence the collision probabilities are preserved. 
    \end{enumerate}
    In the above two operations we introduce $3d$ new generators ($2d$ copies of the identity and $1$ copy of each the original generators). To simplify notation, we assume $S = \{s_1,\ldots, s_d\}$ satisfies the assumptions above and replace $d$ with $4d$ in the final bound.
        
    Let $X_0,X_1,\dots, X_{2t}$ be a simple random walk in $G$ initialized at the identity element, denoted 0.
    Because $\Gam$ is Abelian, the position of $X_t$ at any time can be compressed into 
    the count of the number of times that each generator $s_i$ has been used as a step,
    which we write as the tuple $C^{(t)} \in \N^{d}$.
    The returning walks of length $2t$ are exactly those $c \in \N^{d}$
    such that $\sum_{i = 1}^{d} c_i s_i = 0$ (in $\Gamma$) and $\sum_{i=1}^{d} c_i = 2t$ (in $\N$).
    We have:
    \begin{align*}
        \frac{\CP_t}{\CP_{2t}} = \frac{\Pr(X_{2t} = X_0)}{\Pr(X_{4t} = X_0)} = \frac{\sum_{{c \in \N^{d}:\sum_{i=1}^{d} c_i s_i = 0}} \Pr(C^{(2t)} = c)}{\sum_{{c \in \N^{d}:\sum_{i=1}^{d} c_i s_i = 0}} \Pr(C^{(4t)} = c)}
    \end{align*}
    
    We define $\mu \in \N^{d}$ to be an integer vector whose entries are approximately $\frac{t}{d}$.
    \begin{claim}
        There exists $\mu \in \N^{d}$ such that $\sum_{i=1}^{d} \mu_i = 2t$, $\sum_{i=1}^{d} \mu_i s_i = 0$, and $|\mu_i - \mu_j| \leq 1$ for all $i,j \in [d]$.
    \end{claim}
    \begin{proof}[Proof of claim.]
        Since $S$ is a symmetric set of generators and each generator occurs with even multiplicity, we can pair up the generators with their inverses (for a generator which is its own inverse just pair it with another copy since we assume the number is even). Let $x\in\mathbb{Z}^{d/2}_{\geq 0}$ such that $\sum_{i\in [d/2]}x_i = t$ and for all $i,j\in[d/2]$ we have $|x_i-x_j|\leq 1$. Now for every $r\in [d/2]$ define $\mu_i = \mu_j = x_r$ where $(i,j)$ is the $r$-th pair of generators. It can be verified that $\mu$ satisfies the desired properties.
    \end{proof}
    Let $\mu \in \N^{d}$ be as in the Claim. Then,
    by ignoring terms in the denominator except
    for those with $c_i \geq \mu_i$ for all $i$,
    \begin{align*}
        \frac{\sum_{{c \in \N^{d}:\sum_{i=1}^{d} c_i s_i = 0}} \Pr(C^{(2t)} = c)}{\sum_{{c \in \N^{d}:\sum_{i=1}^{d} c_i s_i = 0}} \Pr(C^{(4t)}) = c)} \leq \frac{\sum_{{c \in \N^{d}:\sum_{i=1}^{d} c_i s_i = 0}} \Pr(C^{(2t)} = c)}{\sum_{{c \in \N^{d}:\sum_{i=1}^{d} c_i s_i = 0}} \Pr(C^{(4t)} = c + \mu)}\,.
    \end{align*}
    The point of the inequality is that it now suffices to show the direct comparison inequality $\frac{\Pr(C^{(2t)} = c)}{\Pr(C^{(4t)} = c + \mu)} \leq (2e)^d$ for all $c \in \N^{d}$ with $\sum_{i = 1}^d c_i = 2t$ (dropping the constraint that $\sum_{i=1}^d c_i s_i=0$ in $\Gamma$).
    Towards this, we have
    \begin{align*}
        \frac{\Pr(C^{(2t)} = c)}{\Pr(C^{(4t)} = c + \mu)} = \frac{\binom{2t}{c_1, \dots, c_d} d^{4t}}{\binom{4t}{c_1 + \mu_1, \dots, c_d + \mu_d} d^{2t}}= \frac{(2t)! (c_1 + \mu_1)!\cdots (c_d + \mu_d)! d^{2t}}{(4t)! c_1!\cdots c_d!}\,.
    \end{align*}
    We prove by ``discrete gradient descent'' that this quantity is maximized when $c = \mu$.
    Let $c'$ be $c$ with $c_i$ replaced
    by $c_i+1$ and $c_j$ replaced by $c_j - 1$. The ratio of the consecutive terms is, 
    \begin{align*}
  \frac{\Pr(C^{(2t)} = c')}{\Pr(C^{(4t)}) = c' + \mu)} \cdot \frac{\Pr(C^{(4t)} = c + \mu)}{\Pr(C^{(2t)} = c)} = \frac{(c_i+\mu_i+1)c_j}{(c_i+1)(c_j+\mu_j)} \,.
    \end{align*}
    This is at least 1 if and only if $\frac{c_j}{c_i+1} \geq \frac{\mu_j}{\mu_i}$. If this holds, the change $(c_i, c_j) \to (c_i + 1, c_j - 1)$ increases the value. This implies that
    $c = \mu$ at the maximizer (since if $c \neq \mu$, there is at least once coordinate which is smaller than $\mu$ and one coordinate which is larger than $\mu$ in which we can move to increase the value).
    
    Finally, we bound the value at the maximizer. 
    \begin{align*}
        \frac{\Pr(C^{(2t)} = \mu)}{\Pr(C^{(4t)} = 2\mu)} &= \frac{(2t)! d^{2t} \prod_{i=1}^d (2\mu_i)!}{(4t)! \prod_{i=1}^d\mu_i!}\\
        &\leq 2^{d/2+1}\frac{\sqrt{2t}(2t/e)^{2t} d^{2t} \prod_{i = 1}^d \sqrt{2\mu_i}(2\mu_i/e)^{2\mu_i}}{\sqrt{4t}(4t/e)^{4t} \prod_{i = 1}^d \sqrt{\mu_i}(\mu_i/e)^{\mu_i}} & (\text{\cref{fact:stirling}})\\
        &= 2^d \cdot \frac{d^{2t} \prod_{i = 1}^d 2^{2\mu_i}\mu_i^{\mu_i}}{2^{4t}(2t)^{2t}}\\
        &\leq 2^d \cdot \frac{d^{2t} \prod_{i = 1}^d 2^{2\mu_i}(\frac{2t}{d}+1)^{\mu_i}}{2^{4t}(2t)^{2t}} & (\mu_i \leq \frac {2t}{d} + 1)\\
        &= 2^d \cdot \left(\frac{d}{2t}\right)^{2t} \left(\frac{2t}{d}+1\right)^{2t} & \left(\sum_{i = 1}^d \mu_i = 2t\right)\\
        &= 2^d \left(1 + \frac d{2t}\right)^{2t} \leq (2e)^d & (1+ x \leq e^x)\,.
    \end{align*}
    Which concludes the proof.
    \end{proof}
\end{lemma}

\section{The sparse cuts of Abelian Cayley graphs}\label{sec:sparsest-cut-live-low-eigenspace}

In this section, we prove that \textit{all} sparse cuts of an Abelian Cayley graph are approximately contained in the low eigenspace with $\tau = O(d\cdot \phi^2)$ thus obtaining \cref{thm:bounded-cut-dimension}.

\begin{theorem}\label{thm:sparsest-cut-low-dim} 
    Let $G=\Cay(\Gamma, S)$ with $\Card{S}=d\,.$
    Let $0<\eps\leq 1$ and $\tau=100d \cdot \phi^2/\eps^2\,.$  For all $Q\subseteq [n], |Q| \leq n/2$ such that $\phi_G(Q)\leq 2\phi(G)\,,$ we have $\Norm{\mathbf \Pi_{\low_\tau} \bar{\mathbf{1}}_Q}^2\geq (1-\eps)\Norm{\bar{\mathbf 1}_Q}^2\,.$
\end{theorem}

The proof extends the combinatorial proof of the Buser inequality in graphs due to Oveis Gharan and Trevisan~\cite{oveis2021ARV}.
Let $Q \subseteq [n]$ be a sparsest cut in $G = \cay(\Gam, S)$.
We analyze the expansion of $Q$ in the graph $G^{2t}$ for an appropriate choice of $t \in \N$.
Following the proof of the Buser inequality~\cite{oveis2021ARV}, this quantity can be bounded in terms of the expansion in $G$. For completeness,
we include the proof of the following Lemma in \cref{sec:buser-alternate}.
\begin{lemma}[\cite{oveis2021ARV}]\label{lem:buser-combinatorial}
    $\phi_{G^{2t}}(Q) \leq 2\sqrt{td} \cdot \phi_G(Q)$.
\end{lemma}

\begin{proof}[Proof of Theorem~\ref{thm:sparsest-cut-low-dim}]
By \cref{fact:rayleigh-quotient} the expansion has a spectral representation,
\begin{equation}\label{eq:phi-gt}\phi_{G^{2t}}(G) = \frac{\mathbf{1}_Q^T \matL(G^{2t}) \mathbf{1}_Q}{\mathbf{1}_Q^T \mathbf{1}_Q}\,.\end{equation}
Let $\mathbf{1}_Q = \sum_{i=1}^n q_i v_i(G)$ be the representation of $\mathbf{1}_Q$ in the eigenbasis. The eigenvalues of $\matL(G^{2t})$ are equal to $1 - (1 - \lam_i(G))^{2t}$. By combining \cref{eq:phi-gt} and  \cref{lem:buser-combinatorial} we obtain,
\[
    \frac{1}{|Q|}\sum_{i=2}^n q_i^2(1-(1-\lam_i)^{2t}) \leq 2\sqrt{td}\cdot \phi_G(Q)\,.
\]

We interpret the left-hand side probabilistically.
Let $i \sim \cS(Q)$ denote the ``spectral sample'' distribution on $\{2,3 ,\dots, n\}$ taking value $i$ with probability proportional to $q_i^2$ i.e. the weight of $\bar{\mathbf{1}}_Q$ on the $i$th eigenvector.
The normalizing factor for $\cS(Q)$ is $\norm{\bar{\mathbf{1}}_Q}^2 = \sum_{i = 2}^n q_i^2 = \frac{|Q|(n - |Q|)}{n} \geq \frac{|Q|}{2}$ using $|Q| \leq n/2$.
Then we have,
\[
    \E_{i \sim \cS(Q)}[1 - e^{-2\lam_i t}] \leq \E_{i \sim \cS(Q)} [1-(1-\lam_i)^{2t}] \leq 8\sqrt{td}\cdot \phi(G)
\]
Fixing a threshold $\tau \geq 0$, we upper bound $\E_{i \sim \cS(Q)}[e^{-2\lam_i t}] \leq (1-p) + pe^{-2\tau t}$ where $p := 1 - \norm{\matPi_{\low_\tau}\bar{\mathbf{1}}_Q}^2/\norm{\bar{\mathbf{1}}_Q}^2$ is the fraction of mass outside of the low eigenspace. Therefore,
\[
    p(1 - e^{-2\tau t}) \leq 8\sqrt{td}\cdot \phi(G)\,.
\]
Selecting $\tau = 100\eps^{-2} d\phi^2(G)$ and $t = 1/\tau\,,$
we conclude $p \leq \eps$ i.e. at least $1-\eps$ fraction
of the mass of $\mathbf{1}_Q$ is on the low eigenspace.
This finishes the proof of \cref{thm:sparsest-cut-low-dim}.
\end{proof}

\subsection{Buser inequality via random walks 
}\label{sec:buser-alternate}

\begin{proof}[Proof of \cref{lem:buser-combinatorial}]
We have assumed that the multiset of generators $S$ is symmetric, meaning that $x$ and $-x$
have the same multiplicity in $S$.
Let $\{s_1, \dots, s_{d'}\} \subseteq S$
be a set of generators ignoring inverses
i.e. we pair up the inverses and take one each of $\{x, -x\}$ and we include all generators which are their own inverse.

We can think of an edge in $G^{2t}$ as a walk of length $2t$ in $G$ which we denote $X_0, X_1, \dots ,X_{2t}$.
Using the fact that we are on an Abelian Cayley graph, the walk can be expressed as $X_{2t} = X_0 + \sum_{i = 1}^{d'} C_i^{(2t)} s_i$ where each $C_i^{(2t)} \in \Z$
is a signed random variable that counts the number of times that generator $s_i$ is used as a step of the walk, using a minus sign when a step is taken on an inverse element.
We initialize the walk at a uniformly random vertex.

The random variable $C_i^{(t)}$ follows a simpler random walk on $\Z$.
It changes by either $\{-1, 0, +1\}$ at each step, the probability
of transitioning to $\pm 1$ is $1/d$,
and due to the symmetry condition of $S$ it has mean zero.
In particular, this walk exhibits a lot of cancellations and we expect it to
have an approximately Gaussian density around 0.

Because of the cancellations,
we can take a shorter walk to reach $X_{2t}$ which only uses $|C_i^{(2t)}|$
of the edges labeled $s_i\,$, say, the first $|C_i^{(2t)}|$ steps in that direction.
Due to the random initialization, each of the steps of the random walk $X_{2t}$ in direction $s_i$ is marginally a uniformly
random edge of the graph in that direction. Let $\partial_i Q = \{(x, x\pm s_i) \in E(G) : x \in Q, x \pm s_i \notin Q\}\,$. By taking a union bound over each of the edges in the shorter walk, we have:
\begin{align*}
    \Pr(X_0 \in Q \wedge X_{2t} \notin Q)
    &\leq \sum_{i=1}^{d'}\E[|C_i^{(2t)}|] \cdot \frac{|\partial_i Q|}{n}\\
    &\leq \sum_{i=1}^{d'}\sqrt{\E[(C_i^{(2t)})^2]} \cdot \frac{|\partial_i Q|}{n}\,.
\end{align*}
Each $C_i^{(2t)}$ is the sum of $2t$ independent, mean-zero random variables that
take values in $\{-1,0,+1\}$ and that are 0 with probability $1-2/d$. We compute $\E[(C_i^{(2t)})^2] = 4t/d$.\footnote{For generators which are their own inverse, we define the walk on $C_i^{(t)}$ to increment either $\pm 1$ at random. The probability of $C_i^{(t)}$ transitioning to $\pm 1$ is $\frac{1}{2d}$.
We compute $\E[(C_i^{(2t)})^2] = 2t/d$ for this case.}
Therefore,
\begin{align*}
        \Pr(X_0 \in Q \wedge X_{2t} \notin Q) &\le \sum_{i=1}^{d'}\sqrt{\frac{4t}{d}} \cdot \frac{|\partial_i Q|}{n}\\
        &= \sqrt{\frac{4t}{d}} \cdot \frac{|\partial Q|}{n}\\
        &= \frac{|Q|}{n}\cdot 2\sqrt{td} \cdot \phi(Q)\,.
\end{align*}
Finally, we have $\Pr(X_0 \in Q \wedge X_{2t}\notin Q) = \frac{|Q|}{n}\cdot \phi_{G^{2t}}(Q)$. Plugging this in completes the claim.
\end{proof}

\section{Cayley graphs with high approximate eigenvalue multiplicity}
In this section we provide two constructions of Abelian Cayley graphs for which the upper bound on eigenvalue multiplicity given by \cref{thm:dimension-low} is tight up to constant factors in the exponent.

The first construction shows that for any finite Abelian group $\Gamma$, there is a Cayley graph $\Cay(\Gamma, S)$ that attains the bound given by \cref{thm:dimension-low}.

\begin{proposition}\label{prop:approx_eig_val_lower_bound}
    For any finite Abelian group $\Gamma$, there exists a subset $S$ with $|S| = \Theta(\log |\Gamma|)$ and constant $c > 0$ such that the Cayley graph $\Cay(\Gamma,S)$ satisfies $\mul_{c}(\lambda_2) = 2^{\Omega(|S|)}$.
\end{proposition}

\cref{prop:approx_eig_val_lower_bound} follows from the classic result of Alon and Roichman \cite{alon1994random} who showed that a random Cayley graph with logarithmic degree is an expander with constant probability. Observe, that $\cay(\Gamma, S)$ is connected then $\mul_{\frac{2}{\lambda_2}}(\lambda_2) = n-1$. If $\cay(\Gamma, S)$ is an expander with degree $|S| = O(\log |\Gamma|)$ we can upper bound $2/\lambda_2$ by some constant $c$. This implies, $\mul_{\frac{2}{\lambda_2}}(\lambda_2) = n-1\leq 2^{\Omega(|S|)}$.

For the second construction, we use the relationship between binary linear codes and Cayley graphs over $\Z^k_2$ to construct Abelian Cayley graphs with large eigenvalue multiplicity. This construction yields the following lower bound, matching \cref{thm:dimension-low}. See~\cref{app:codes}
for the proof. 
\begin{proposition}\label{prop:eig_val_lower_bound}
   Let $n= 2^k$. There is a family of Cayley graphs $\Cay(\Z^k_2, S)$ such that $|S| = \Theta(\log n)$ and $\mul_{\lambda_2}\geq 2^{\Omega(|S|)} = n^{\Omega(1)}\,$.
\end{proposition}

Both constructions in this section are expanders and have logarithmic degree. A natural question is whether our bound on approximate eigenvalue multiplicity is tight for non-expanding Abelian Cayley graphs with sub-logarithmic degree.

\begin{question}
    Does there exist Abelian Cayley graphs $\Cay(G,S)$ with degree $|S| = o(\log |G|)$ and $\lambda_2 = o(1)$ such that $\mul_c(\lambda_2) \geq 2^{\Omega(|S|)}$?
\end{question}

\subsection{Eigenvalue multiplicity and binary linear codes}
\label{app:codes}
We begin with some basic definitions from coding theory found in \cite{guruswami2019essential}.

\begin{definition}[Binary Linear Code]
   A binary linear code $C$ of dimension $k$ and block length $n$ is a $k$-dimensional linear subspace of $\Z^n_2$. 
\end{definition}

The \emph{distance} between two elements $x,y\in \Z^n_2$, denoted by $\Delta(x,y)$ is the number of positions in which $x$ and $y$ differ. The \emph{relative distance} between $x$ and $y$ is $\delta(x,y) = \Delta(x,y)/n$.
 The \emph{distance} of a code $C$ is $\Delta(C) = \min_{c\neq c'}\Delta(c,c')$ and the \emph{relative distance} is $\delta(C) = \Delta(C)/n$. Linear codes have the nice property that the distance can be rewritten as $\Delta(C) = \min_{c\neq 0}|c|$, where $|c| = \Delta(c,0)$ is the \emph{Hamming weight} of $c$. The \emph{relative Hamming weight} of a vector $c$ is $|c|/n$.
 
 We write elements of $\Z^n_2$ as row vectors. A \emph{generator matrix} for the code $C$ is a rank $k$ matrix $\matG\in\Z^{k\times n}_2$ whose row span is $C$. 
Let $S\subseteq \Z^k_2$ be a set of size $n$ and consider the Cayley graph $\Cay(\Z^k_2, S)$. We can define a binary linear code with generator matrix
\[\matG_S=
\left[
  \begin{array}{cccc}
    \vertbar & \vertbar &        & \vertbar \\
    s_{1}    & s_{2}    & \ldots & s_{n}    \\
    \vertbar & \vertbar &        & \vertbar 
  \end{array}
\right].
\]

Viewing $\matG_S$ as a linear map from $\Z^k_2$ to $\Z^n_2$ we have that $m\matG_S = (\langle s_1, m\rangle, \ldots, \langle s_n, m\rangle)$. The code generated by $\matG_S$ is $C_S=\mathsf{Im}(\matG_S)\subseteq \Z^n_2$. The relationship between $\Cay(\Z^k_2, S)$ and $C_S$ is summarized by the following well-known fact,
which is a consequence of the eigenvectors for $\Z_2^k$ being the Boolean Fourier characters. 

\begin{proposition}\label{fact:blc_distance}
Let $S \subseteq \Z_2^k$ such that $\matG_S$ has rank $k$.
Let $\lambda_2$ be the second smallest normalized Laplacian eigenvalue of $\Cay(\Z^k_2,S)$. Then $\lambda_2/2 = \delta(C_S)$. Furthermore, the eigenvalue multiplicity of $\lambda_2$ is equal to the number of code words of minimum weight in $C_S$.
\end{proposition}

The assumption $\rank(\matG_S) = k$ is equivalent to $\Cay(\Z_2^k, S)$ being connected.

\cref{fact:blc_distance} shows that the maximum eigenvalue multiplicity of Cayley graphs over $\Z^k_2$ is equal to the maximum number of minimum-weight code words in binary linear codes.
The code version of the question has been studied by Ashikhmin, Barg and Vl\u{a}du\c{t} \cite{ashikhmin2001linear} answering a question of Kalai and Linial \cite{kalai1995distance}.

For a code $C$ define $A_{\Delta} = \{x\in C: |x| = \Delta(C) \}$ to be the set of non-zero code words of minimum weight. Define $E_q(\delta) = H(\delta) - \frac{\log q}{\sqrt{q}-1}-\log\frac{q}{q-1}$. For $q\geq 49\,,$ $E_q(\delta)$ has two roots $0 < \delta_1(q) <\delta_2(q)$ and is positive for all $\delta_1(q) < \delta < \delta_2(q)$. Now we can state the main result of \cite{ashikhmin2001linear}.

\begin{theorem}\label{thm:codes_many_light_codewords}\cite{ashikhmin2001linear}
   Fix $s \in \N\,, q = 2^{2s}$ such that $q \geq 49$. Then for any $\delta_1(q) < \delta < \delta_2(q)$ there exists a sequence $k \to \infty$ and a binary linear code $C$ of dimension $k\,,$ block length $n = qk\,,$ and distance $\Delta(C) \geq n\delta/2$ such that
   \[
   \log \abs{A_{\Delta}} \geq kE_{q}(\delta)-o(k).
   \]
\end{theorem}

The proof of \cref{prop:eig_val_lower_bound} follows from combining \cref{thm:codes_many_light_codewords} and \cref{fact:blc_distance}.

\begin{proof}[Proof of \cref{prop:eig_val_lower_bound}]
\cref{thm:codes_many_light_codewords} yields a code in which the number of minimum-weight codewords is at least $2^{\Omega(k)}$ out of the total number of codewords $2^k$.
We can convert this into a Cayley graph $\Cay(\Z_2^k, S)$ with $\mul_{\lambda_2} \geq 2^{\Omega(k)}$ by \cref{fact:blc_distance}. The degree of the Cayley graph is $|S| = n = \Theta(k)$ which is logarithmic in the number of vertices of the graph. Hence, $\mul_{\lambda_2} \geq 2^{\Omega(k)} = 2^{\Omega(|S|)}$.
\end{proof}

\section{Sparsest cut in \texorpdfstring{$\mathbb{Z}^k_p$}{Znp}}\label{app:sparsest-cut-vector-space}

\sparsestcut over the Boolean hypercube $\Z_2^k$ can be solved exactly by a different type of spectral algorithm. The eigenvectors of $\Z^k_2$ are $\pm 1$-valued (the Boolean Fourier characters) which implies that (1)~the lower bound in the Cheeger inequality is exact $\tfrac 12 \lam_2 = \phi(G)\,$, (2)~there is an eigenvector which sign-indicates a sparsest cut.

In this section, we show that there is a simple polynomial time algorithm that computes a $O(p)$-approximation for sparsest cut on undirected Cayley graphs over $\Z^k_p$, where $p$ is prime.

\begin{proposition}
\torestate{\label{prop:sparse_cut_vector_space}
    There exists a polynomial time algorithm that computes an $O(p)$-approximation on undirected Cayley graphs over $\Z^k_p$. 
    }
\end{proposition}

The idea is given $G = \cay(\Z^k_p, S)$, we define a related graph $G' = \cay(\Z^k_p, S')$ such that 
\begin{enumerate}
    \item $\phi(G')$ is a $O(p)$-approximation of $\phi(G)$, and
    \item $\lambda_2(G')$ is a $O(1)$-approximation of $\phi(G')$.
\end{enumerate}
Combining both $(1)$ and $(2)$ implies $\lambda_2(G')$ is an $O(p)$-approximation of $\phi(G)$. Hence, the polynomial time $O(p)$-approximation algorithm for sparsest cut on $\Z^k_p$ is computing $\lambda_2(G')$. In fact, our proof shows that the subspace orthogonal to $g$ where $\chi_{g}$ is an eigenvector corresponding to $\lambda_2(G')$ provides an $O(p)$-approximation to $\phi(G)$.

Let $G = \cay(\Z^k_p, S)$. Consider $G' = \cay(\Z^k_p, S')$, where $S'$ is the multiset of size $\frac{p-1}{2}|S|$ containing $\cup_{\ell\in[(p-1)/2]}\ell S$, where $\ell S = \{\ell s:s\in S\}$. The set $S'$ is obtained from $S$ by taking all non-zero multiples of $S$ up to $(p-1)/2$. By symmetry of $S$, the multiset $S'$ contains all non-zero multiples of each $x\in S$. For each $\ell \neq 0$, the graph $G_\ell = \cay(\Z^k_p, \ell S)$ is isomorphic to $\cay(\Z^k_p, S)$. In this section we use the fact that the eigenvectors of Cayley graphs over $\Z^k_p$ are given by $\chi_g(x) = e^{2\pi i\langle x, g \rangle/p}$, where $g\in \Z^k_p$ (see chapter $16$ of \cite{trevisan2017lecture}).

The lemma below shows that $\phi(G)$ and $\phi(G')$ differ by at most a factor of $(p+1)/4$. 
\begin{lemma}\label{lem:p_approx_cond}
    The graphs $G$ and $G'$ satisfy that,
    \[
    \phi(G)\leq \phi(G')\leq \frac{p+1}{4}\phi(G).
    \]
\end{lemma}

\begin{proof}
    We begin with the lower bound $\phi(G)\leq \phi(G')$. Let $Q\subseteq \Z^k_p$ such that $|Q|\leq p^k/2$. Observe that $E_{G'}(Q, \bar{Q}) = \sum_{\ell\in[(p-1)/2]}E_{G_\ell}(Q, \bar{Q})$. Since $G_\ell$ is isomorphic to $G$ for all $\ell\in[(p-1)/2]$, we have $E_{G_\ell}(Q, \bar{Q})\geq \phi(G)|S||Q|$. This implies
    \begin{align*}
        \phi_{G'}(Q) &= \frac{E_{G'}(Q, \bar{Q})}{|S'||Q|}\\
        &= \frac{\sum_{\ell\in[(p-1)/2]}E_{G_\ell}(Q,\bar{Q})}{|S'||Q|}\\
        &\geq \frac{\frac{p-1}{2}|S||Q|\phi(G)}{|S'||Q|}\\
        &= \phi(G).
    \end{align*}
   Hence, $\phi(G')\geq \phi(G)$. 
   
   Now, we prove the upper bound $\phi(G')\leq \frac{p+1}{4}\phi(G)$. Observe that we can write the number of edges cut by $Q$ as $E_G(Q,\bar{Q}) = \frac{1}{2}\sum_{s\in S}\sum_{x\in \Z^k_p}(\mathbf{1}_Q(x)-\mathbf{1}_Q(x+s))^2$ and $E_{G'}(Q, \bar{Q}) = \frac{1}{2}\sum_{\ell\in[(p-1)/2]}\sum_{s\in S}\sum_{x\in \Z^k_p}(\mathbf{1}_Q(x)-\mathbf{1}_Q(x+\ell s))^2$. The terms $(\mathbf{1}_Q(x)-\mathbf{1}_Q(x+\ell s))^2$ satisfy a ``triangle inequality''  
   \[(\mathbf{1}_Q(x)-\mathbf{1}_Q(x+\ell s))^2\leq \sum_{i=1}^\ell(\mathbf{1}_Q(x+(i-1)s)-\mathbf{1}_Q(x+is))^2.\]
   
   Note that for each fixed $i\in[\ell]$, $\sum_{x\in\Z^k_p}(\mathbf{1}_Q(x+(i-1)s)-\mathbf{1}_Q(x+is))^2= \sum_{x\in\Z^k_p}(\mathbf{1}_Q(x)-\mathbf{1}_Q(x+s))^2$. This implies, 
   \begin{align*}
       \sum_{x\in\Z^k_p}(\mathbf{1}_Q(x)-\mathbf{1}_Q(x+\ell s))^2&\leq \sum_{x\in\Z^k_p}\sum_{i=1}^\ell(\mathbf{1}_Q(x+(i-1)s)-\mathbf{1}_Q(x+is))^2\\
       &\leq \ell\sum_{x}(\mathbf{1}_Q(x)-\mathbf{1}_Q(x+s))^2.
   \end{align*}
Using this we obtain the following bound on $E_{G'}(Q,\bar{Q})$ in terms of $E_{G}(Q,\bar{Q})$
\begin{align*}
  E_{G'}(Q,\bar{Q}) &= \frac{1}{2}\sum_{\ell\in[(p-1)/2]}\sum_{s\in S}\sum_{x\in \Z^p_n}(\mathbf{1}_Q(x)-\mathbf{1}_Q(x+\ell s))^2\\
  &\leq \frac{1}{2}\sum_{\ell\in[(p-1)/2]}\sum_{s\in S}\sum_{x\in\Z^k_p}(\mathbf{1}_Q(x)-\mathbf{1}_Q(x+s))^2\\
  &=\frac{(p-1)(p+1)}{8}E_G(Q,\bar{Q}).
\end{align*}
Dividing by $|S'||Q|$, we obtain $\phi_{G'}(Q)\leq \frac{p+1}{4}\phi(Q)$. Thus, $\phi(G')\leq \frac{p+1}{4}\phi(G)$, as desired.
\end{proof}

To conclude the analysis, we show that $\lambda_2(G')$ and $\phi(G')$ differ by a factor of at most $1/2$. Our proof makes use of a well-known lemma which can be found in \cite{hoory2006expander}. 

\begin{lemma}\cite{hoory2006expander}\label{lem:hlw}
    Let $G = (V,E)$ be a graph, $\matL$ the normalized Laplacian of $G$, and $\gamma$ an eigenvector corresponding to $\lambda_2(G)$. Define $\gamma_+$ by $\gamma_+(x) = \max\{f(x), 0\}$ and $\gamma_-$ by $\gamma_-(x) = \min\{\gamma(x), 0\}$. Then $\gamma_+,\gamma_-$ are disjointly supported and satisfy
    \[
    \frac{\gamma^T_{-} \matL_{G'}\gamma_-}{\gamma^T_-\gamma_-} \leq \lam_2(G), \qquad  \frac{\gamma^T_{+} \matL_{G'}\gamma_+}{\gamma^T_+\gamma_+}\leq \lambda_2(G).
    \]
\end{lemma}

\begin{lemma}\label{lem:lambda_approx_cond}
 The graph $G'$ satisfies 
    \[
    \frac{\lambda_2(G')}{2}\leq \phi(G')\leq \lambda_2(G').
    \]
    
\end{lemma}

\begin{proof}
    The first inequality is a direct application of Cheeger's inequality. It remains to prove the second inequality $\phi(G')\leq \lambda_2(G')$.
    Let $\ell \neq 0$ and $g\in\Z^k_p$ and $\chi_g$ be the corresponding eigenvector. We claim that if $\chi_{g}$ is an eigenvector corresponding to $\lambda_2(G')$, then so is $\chi_{\ell g}$. By definition, for every $s\in S'$ the element $\ell s\in S'$ occurs with the same multiplicity. This implies $\sum_{s\in S'}\chi_{g}(s) = \sum_{s\in S'}\chi_{\ell g}(s)$. Hence, $\chi_{\ell g}$ is also an eigenvector corresponding to $\lambda_2(G')$.
    
    Consider the symmetrized eigenvector $\gamma = \sum_{\ell \in[p-1]}\chi_{\ell g}(s)$ corresponding to $\lambda_2(G')$. One can see that for each $x\in \Z^k_p$, the vector $\gamma$ satisfies
    \[\gamma(x) = 
    \begin{cases}
        p-1 &\text{if } \langle g, x\rangle = 0\\
        -1 &\text{otherwise.}
    \end{cases}
    \]
    
Define $\gamma_+$ by $\gamma_+(x) = \max\{\gamma(x), 0\}$. Observe that $\gamma_+ = (p-1)\mathbf{1}_Q$, where $Q = \{x\in \Z^k_p: \langle g, x\rangle = 0\}$ is the subspace orthogonal to $g$. Applying Lemma~\ref{lem:hlw}, we obtain
    \[
    \phi(G')\leq \phi_{G'}(Q) = \frac{\mathbf{1}^T_Q\matL_{G'} \mathbf{1}_Q}{\mathbf{1}^T_Q\mathbf{1}_Q} = \frac{\gamma^T_{+} \matL_{G'}\gamma_+}{\gamma^T_+\gamma_+} \leq \lambda_2(G').
    \]
\end{proof}

Combining both lemmas gives us proves Proposition~\ref{prop:sparse_cut_vector_space}.

\begin{proof}[Proof of Proposition~\ref{prop:sparse_cut_vector_space}]
Let $G = \cay(\Z^k_p, S)$ and $G' = \cay(\Z^k_p, S')$, where $S'$ is the multiset of size $\frac{p-1}{2}|S|$ containing $\cup_{\ell\in[(p-1)/2]}kS$, where $\ell S = \{\ell s:s\in S\}$. One can see that the graph $G'$ and the eigenvalue $\lambda_2(G')$ can be computed in time polynomial in $|\Z^k_p|$. Combining Lemma~\ref{lem:p_approx_cond} and Lemma~\ref{lem:lambda_approx_cond} gives us 
\[
\phi(G)\leq \lambda_2(G')\leq \frac{p+1}{2}\phi(G).
\]
Hence, $\lambda_2(G')$ provides a $O(p)$-approximation to $\phi(G)$.
\end{proof}

\ifnum\conferenceversion=0
\subsection*{Acknowledgments}
We deeply thank Luca Trevisan for his motivation and wisdom on this problem, and for suggesting to look into Abelian Cayley graphs.
CJ and JZ thank Lucas Pesenti and Robert Wang for discussions. We thank Madhu Sudan for pointing us to \cite{ashikhmin2001linear}.

JR and SV are supported in part by a Simons Investigator Award to Salil Vadhan. Work began while JR and SV were visitors at Bocconi University. SV was a Visiting Researcher at the Bocconi University Department of Computing Sciences, supported by Luca Trevisan's ERC Project GA-834861.
CJ is supported in part by the European Research Council (ERC) under the European Union’s Horizon 2020 research and innovation programme (grant agreement Nos. 834861 and 101019547).
CJ is also a member of the Bocconi Institute for Data Science and Analytics (BIDSA).

\fi

\phantomsection
\addcontentsline{toc}{section}{Bibliography}
{\footnotesize
\bibliographystyle{amsalpha} 
\bibliography{refs, scholar}

\newcommand{\etalchar}[1]{$^{#1}$}
\providecommand{\bysame}{\leavevmode\hbox to3em{\hrulefill}\thinspace}
\providecommand{\MR}{\relax\ifhmode\unskip\space\fi MR }
\providecommand{\MRhref}[2]{%
  \href{http://www.ams.org/mathscinet-getitem?mr=#1}{#2}
}
\providecommand{\href}[2]{#2}
\begin{thebibliography}{AKK{\etalchar{+}}08}

\bibitem[ABS15]{arora2015subexponential}
Sanjeev Arora, Boaz Barak, and David Steurer, \emph{Subexponential algorithms
  for unique games and related problems}, Journal of the ACM (JACM) \textbf{62}
  (2015), no.~5, 1--25.

\bibitem[ABV01]{ashikhmin2001linear}
Alexei Ashikhmin, Alexander Barg, and Serge Vladut, \emph{Linear codes with
  exponentially many light vectors}, Journal of combinatorial theory. Series A
  \textbf{96} (2001), no.~2, 396--399.

\bibitem[AKK{\etalchar{+}}08]{arora2008unique}
Sanjeev Arora, Subhash~A Khot, Alexandra Kolla, David Steurer, Madhur Tulsiani,
  and Nisheeth~K Vishnoi, \emph{Unique games on expanding constraint graphs are
  easy}, Proceedings of the fortieth annual ACM symposium on Theory of
  computing, 2008, pp.~21--28.

\bibitem[AL08]{andersen2008algorithm}
Reid Andersen and Kevin~J Lang, \emph{An algorithm for improving graph
  partitions.}, SODA, vol.~8, 2008, pp.~651--660.

\bibitem[AR94]{alon1994random}
Noga Alon and Yuval Roichman, \emph{Random cayley graphs and expanders}, Random
  Structures \& Algorithms \textbf{5} (1994), no.~2, 271--284.

\bibitem[ARV09]{arora2009expander}
Sanjeev Arora, Satish Rao, and Umesh Vazirani, \emph{Expander flows, geometric
  embeddings and graph partitioning}, Journal of the ACM (JACM) \textbf{56}
  (2009), no.~2, 1--37.

\bibitem[BBK{\etalchar{+}}21]{bafna2021playing}
Mitali Bafna, Boaz Barak, Pravesh~K Kothari, Tselil Schramm, and David Steurer,
  \emph{Playing unique games on certified small-set expanders}, Proceedings of
  the 53rd Annual ACM SIGACT Symposium on Theory of Computing, 2021,
  pp.~1629--1642.

\bibitem[BHKL22]{bafna2022high}
Mitali Bafna, Max Hopkins, Tali Kaufman, and Shachar Lovett, \emph{High
  dimensional expanders: Eigenstripping, pseudorandomness, and unique games},
  Proceedings of the 2022 Annual ACM-SIAM Symposium on Discrete Algorithms
  (SODA), SIAM, 2022, pp.~1069--1128.

\bibitem[BM23]{bafna2023solving}
Mitali Bafna and Dor Minzer, \emph{Solving unique games over globally
  hypercontractive graphs}, arXiv preprint arXiv:2304.07284 (2023).

\bibitem[BRS11]{barak2011rounding}
Boaz Barak, Prasad Raghavendra, and David Steurer, \emph{Rounding semidefinite
  programming hierarchies via global correlation}, 2011 ieee 52nd annual
  symposium on foundations of computer science, IEEE, 2011, pp.~472--481.

\bibitem[BRT21]{benson2021volume}
Brian Benson, Peter Ralli, and Prasad Tetali, \emph{Volume growth, curvature,
  and buser-type inequalities in graphs}, International Mathematics Research
  Notices \textbf{2021} (2021), no.~22, 17091--17139.

\bibitem[BT16]{breuillard2016nilprogressions}
Emmanuel Breuillard and Matthew~CH Tointon, \emph{Nilprogressions and groups
  with moderate growth}, Advances in Mathematics \textbf{289} (2016),
  1008--1055.

\bibitem[CKK{\etalchar{+}}06]{chawla2006hardness}
Shuchi Chawla, Robert Krauthgamer, Ravi Kumar, Yuval Rabani, and D~Sivakumar,
  \emph{On the hardness of approximating multicut and sparsest-cut},
  computational complexity \textbf{15} (2006), 94--114.

\bibitem[CKK{\etalchar{+}}21]{curvature_ricci_flatness}
David Cushing, Supanat Kamtue, Riikka Kangaslampi, Shiping Liu, and Norbert
  Peyerimhoff, \emph{Curvatures, graph products and ricci flatness}, Journal of
  Graph Theory \textbf{96} (2021), no.~4, 522--553.

\bibitem[CM97]{colding1997harmonic}
Tobias~H Colding and William~P Minicozzi, \emph{Harmonic functions on
  manifolds}, Annals of mathematics \textbf{146} (1997), no.~3, 725--747.

\bibitem[DSC94]{diaconis1994moderate}
Persi Diaconis and Laurent Saloff-Coste, \emph{Moderate growth and random walk
  on finite groups}, Geometric \& Functional Analysis GAFA \textbf{4} (1994),
  1--36.

\bibitem[Fie73]{fiedler1973algebraic}
Miroslav Fiedler, \emph{Algebraic connectivity of graphs}, Czechoslovak
  mathematical journal \textbf{23} (1973), no.~2, 298--305.

\bibitem[FLGM23]{fountoulakis2023flow}
Kimon Fountoulakis, Meng Liu, David~F Gleich, and Michael~W Mahoney,
  \emph{Flow-based algorithms for improving clusters: A unifying framework,
  software, and performance}, SIAM Review \textbf{65} (2023), no.~1, 59--143.

\bibitem[FMT06]{friedman2006spectral}
Joel Friedman, Ram Murty, and Jean-Pierre Tillich, \emph{Spectral estimates for
  abelian cayley graphs}, Journal of Combinatorial Theory, Series B \textbf{96}
  (2006), no.~1, 111--121.

\bibitem[GKL03]{gupta2003bounded}
Anupam Gupta, Robert Krauthgamer, and James~R Lee, \emph{Bounded geometries,
  fractals, and low-distortion embeddings}, 44th Annual IEEE Symposium on
  Foundations of Computer Science, 2003. Proceedings., IEEE, 2003,
  pp.~534--543.

\bibitem[GM98]{guattery1998quality}
Stephen Guattery and Gary~L Miller, \emph{On the quality of spectral
  separators}, SIAM Journal on Matrix Analysis and Applications \textbf{19}
  (1998), no.~3, 701--719.

\bibitem[GRS19]{guruswami2019essential}
Venkatesan Guruswami, Atri Rudra, and Madhu Sudan, \emph{Essential coding
  theory}, Draft available at http://cse. buffalo.
  edu/faculty/atri/courses/coding-theory/book (2019).

\bibitem[GS13]{guruswami2013approximating}
Venkatesan Guruswami and Ali~Kemal Sinop, \emph{Approximating non-uniform
  sparsest cut via generalized spectra}, Proceedings of the twenty-fourth
  annual ACM-SIAM symposium on Discrete algorithms, SIAM, 2013, pp.~295--305.

\bibitem[GW95]{goemans1995improved}
Michel~X Goemans and David~P Williamson, \emph{Improved approximation
  algorithms for maximum cut and satisfiability problems using semidefinite
  programming}, Journal of the ACM (JACM) \textbf{42} (1995), no.~6,
  1115--1145.

\bibitem[HLW06]{hoory2006expander}
Shlomo Hoory, Nathan Linial, and Avi Wigderson, \emph{Expander graphs and their
  applications}, Bulletin of the American Mathematical Society \textbf{43}
  (2006), no.~4, 439--561.

\bibitem[JTY{\etalchar{+}}21]{jiang2021equiangular}
Zilin Jiang, Jonathan Tidor, Yuan Yao, Shengtong Zhang, and Yufei Zhao,
  \emph{Equiangular lines with a fixed angle}, Annals of Mathematics
  \textbf{194} (2021), no.~3, 729--743.

\bibitem[JTY{\etalchar{+}}23]{jiang2023spherical}
\bysame, \emph{Spherical two-distance sets and eigenvalues of signed graphs},
  Combinatorica \textbf{43} (2023), no.~2, 203--232.

\bibitem[Kho02]{khot2002power}
Subhash Khot, \emph{On the power of unique 2-prover 1-round games}, Proceedings
  of the thiry-fourth annual ACM symposium on Theory of computing, 2002,
  pp.~767--775.

\bibitem[KKRT16]{klartag2016discrete}
Bo'az Klartag, Gady Kozma, Peter Ralli, and Prasad Tetali, \emph{Discrete
  curvature and abelian groups}, Canadian Journal of Mathematics \textbf{68}
  (2016), no.~3, 655--674.

\bibitem[KL95]{kalai1995distance}
Gil Kalai and Nathan Linial, \emph{On the distance distribution of codes}, IEEE
  Transactions on Information Theory \textbf{41} (1995), no.~5, 1467--1472.

\bibitem[Kla81]{klawe1981non}
Maria Klawe, \emph{Non-existence of one-dimensional expanding graphs}, 22nd
  Annual Symposium on Foundations of Computer Science (sfcs 1981), IEEE, 1981,
  pp.~109--114.

\bibitem[Kle10]{kleiner2010new}
Bruce Kleiner, \emph{A new proof of gromov’s theorem on groups of polynomial
  growth}, Journal of the American Mathematical Society \textbf{23} (2010),
  no.~3, 815--829.

\bibitem[KLL{\etalchar{+}}13]{improved_cheeger}
Tsz~Chiu Kwok, Lap~Chi Lau, Yin~Tat Lee, Shayan Oveis~Gharan, and Luca
  Trevisan, \emph{Improved cheeger's inequality: analysis of spectral
  partitioning algorithms through higher order spectral gap}, Proceedings of
  the Forty-Fifth Annual ACM Symposium on Theory of Computing (New York, NY,
  USA), STOC '13, Association for Computing Machinery, 2013, p.~11–20.

\bibitem[KV15]{khot2015unique}
Subhash~A Khot and Nisheeth~K Vishnoi, \emph{The unique games conjecture,
  integrality gap for cut problems and embeddability of negative-type metrics
  into $\ell$-1}, Journal of the ACM (JACM) \textbf{62} (2015), no.~1, 1--39.

\bibitem[LM08]{lee2008eigenvalue}
James~R Lee and Yury Makarychev, \emph{Eigenvalue multiplicity and volume
  growth}, arXiv preprint arXiv:0806.1745 (2008).

\bibitem[LMO09]{lang2009empirical}
Kevin~J Lang, Michael~W Mahoney, and Lorenzo Orecchia, \emph{Empirical
  evaluation of graph partitioning using spectral embeddings and flow},
  International Symposium on Experimental Algorithms, Springer, 2009,
  pp.~197--208.

\bibitem[LOGT14]{lee2014multiway}
James~R Lee, Shayan Oveis~Gharan, and Luca Trevisan, \emph{Multiway spectral
  partitioning and higher-order cheeger inequalities}, Journal of the ACM
  (JACM) \textbf{61} (2014), no.~6, 1--30.

\bibitem[LR99]{leighton1999multicommodity}
Tom Leighton and Satish Rao, \emph{Multicommodity max-flow min-cut theorems and
  their use in designing approximation algorithms}, Journal of the ACM (JACM)
  \textbf{46} (1999), no.~6, 787--832.

\bibitem[LRTV12]{louis2012many}
Anand Louis, Prasad Raghavendra, Prasad Tetali, and Santosh Vempala, \emph{Many
  sparse cuts via higher eigenvalues}, Proceedings of the forty-fourth annual
  ACM symposium on Theory of computing, 2012, pp.~1131--1140.

\bibitem[Mar73]{margulis1973explicit}
Grigorii~Aleksandrovich Margulis, \emph{Explicit constructions of
  concentrators}, Problemy Peredachi Informatsii \textbf{9} (1973), no.~4,
  71--80.

\bibitem[MOV09]{mahoney2009spectral}
Michael~W Mahoney, Lorenzo Orecchia, and Nisheeth~K Vishnoi, \emph{A spectral
  algorithm for improving graph partitions}, Tech. report, Technical report.
  Preprint, 2009.

\bibitem[MRS21]{mckenzie2021support}
Theo McKenzie, Peter Michael~Reichstein Rasmussen, and Nikhil Srivastava,
  \emph{Support of closed walks and second eigenvalue multiplicity of graphs},
  Proceedings of the 53rd Annual ACM SIGACT Symposium on Theory of Computing,
  2021, pp.~396--407.

\bibitem[M{\"u}n23]{munch2023non}
Florentin M{\"u}nch, \emph{Non-negative ollivier curvature on graphs, reverse
  poincar{\'e} inequality, buser inequality, liouville property, harnack
  inequality and eigenvalue estimates}, Journal de Math{\'e}matiques Pures et
  Appliqu{\'e}es \textbf{170} (2023), 231--257.

\bibitem[OT21]{oveis2021ARV}
Shayan {Oveis Gharan} and Luca Trevisan, \emph{{ARV} on abelian cayley graphs},
  In Theory blog (2021).

\bibitem[PS22]{pak2022algebraic}
Igor Pak and David Soukup, \emph{Algebraic and arithmetic properties of the
  cogrowth sequence of nilpotent groups}, arXiv preprint arXiv:2210.09419
  (2022).

\bibitem[RS10]{raghavendra2010graph}
Prasad Raghavendra and David Steurer, \emph{Graph expansion and the unique
  games conjecture}, Proceedings of the forty-second ACM symposium on Theory of
  computing, 2010, pp.~755--764.

\bibitem[RST12]{raghavendra2012reductions}
Prasad Raghavendra, David Steurer, and Madhur Tulsiani, \emph{Reductions
  between expansion problems}, 2012 IEEE 27th Conference on Computational
  Complexity, IEEE, 2012, pp.~64--73.

\bibitem[Tre17]{trevisan2017lecture}
Luca Trevisan, \emph{Lecture notes on graph partitioning, expanders and
  spectral methods}, University of California, Berkeley,
  https://lucatrevisan.github.io/books/expanders-2016.pdf (2017).

\bibitem[Wik]{stirlingWikipedia}
Wikipedia, \emph{Stirling's approximation},
  \url{https://en.wikipedia.org/wiki/Stirling\%27s\_approximation}, Accessed:
  2024-10-24.

\end{thebibliography}
}
\clearpage

\appendix
\section{Low threshold rank implies graph decomposition}\label{app:threshold-rank}

\begin{proposition}
    Let $G$ be a graph and $2 \geq \tau \geq 0$.
    If $\mul_\tau(G) = r\,,$ then $V(G)$ has a partition into $H_1, \dots, H_{r'}$ with $r' \leq r+1$ such that: 
    \begin{enumerate}[(i)]
        \item $\phi_G(H_i) \leq \tau$ for all $i = 1, \dots, r'$ except possibly $i = r'$
        \item $\phi_G(K) > \tau$ for all $i = 1,\dots, r'$ and subgraphs $K \subset H_i\,.$
    \end{enumerate}
\end{proposition}
\begin{proof}
    Maintain a set $V$ which is initially equal to $V(G)\,$.
    Select the smallest subgraph $H \subseteq V$ such that $\phi_G(H) \leq \tau\,,$ if any. We output $H$ as a piece of the partition, then remove it and recursively proceed on $V \gets V \setminus H$.
    We take the final piece of the partition to be $V$ once the process terminates.

    It is clear that this generates a partition of $V(G)\,,$ and that conditions (i) and (ii) are satisfied.
    It remains to show $r' \leq r+1$.

    Each of the subgraphs $H_i$ except for the last one has expansion at most $\tau\,.$ Therefore, by \cref{fact:rayleigh-quotient}, we have a collection of $r' - 1$ disjointly-supported (thus orthogonal) vectors $\matD^{1/2}\mathbf{1}_{H_i}$ whose Rayleigh quotients are at most $\tau$.
    This implies $\lam_{r'-1}(G) \leq \tau\,,$
    as can be proven using the variational formula for the eigenvalues:
    \[
        \lam_k(G) = \min_{\substack{v_1, \dots, v_k \in \R^n:\\v_i \perp v_j}} \;\max_{v \in \linspan(v_1, \dots, v_k)} \frac{v^T \matL(G) v}{v^T v}\,.
    \]
    We conclude $r' - 1 \leq r$.
\end{proof}

\chris{TODO make constructive. Ne need to prove that our cut improvement algorithm satisfies $\hat Q \approx Q$ in order to do that.}

\section{Solution dimension and threshold rank}

\subsection{Solution dimension of the cycle graph}\label{app:cycle-graph}

The fact that the solution dimension can be smaller than the $\phi$-threshold-rank is unexpected because every cut must have some component in the eigenspace to at least $\phi$.
This can be seen using the spectral interpretation of the conductance in \cref{fact:rayleigh-quotient}, which says that for all $Q \subseteq [n]$\,, assuming $G$ is regular,
\begin{align}
    \phi_G(Q) = \frac{\mathbf{1}_Q^T \matL \mathbf{1}_Q}{\mathbf{1}_Q^T \mathbf{1}_Q} &= \E_{i \sim \cS(Q)}[\lam_i] \label{eq:avg-eval}
\end{align}
where $\cS(Q)$ is the distribution on $[n]$ proportional to $\iprod{\mathbf{1}_Q, v_i}^2$ and $v_i$ is an eigenbasis for $G$.

\cref{eq:avg-eval} is at least the expansion parameter $\phi$, implying that every Boolean vector $\mathbf{1}_Q$
must have some component in the eigenspaces to at least $\phi$.
In other words, $\CD_0(G) \geq \mul_{\phi}(G)\,,$
and in order to exactly recover a cut indicator, it is necessary to explore
the entire eigenspace up to at least $\phi$.
Conversely, Markov's inequality on \cref{eq:avg-eval} shows that every cut satisfying $\phi_G(Q) \leq 2 \phi(G)$ must have all but $\eps$ fraction of its weight on eigenvalues up to $2\eps^{-1}\phi$. 

Surprisingly, it may still be that only a tiny fraction of $\mathbf{1}_Q$ lies in the eigenspaces above $\phi$, \ie $\CD_\eps(G) \ll \mul_{\phi}(G)$.

For example, the cycle graph $G = \Cay(\Z_n, \{\pm 1\})$ has $\phi = \Theta(\frac{1}{n})$
and the sparsest cuts are the bisections of the cycle into two halves.
An eigenbasis for the cycle consists of the Fourier characters,
\begin{align*}
     v_\al &: \Z_n \to \C\,, \\
     v_\al(x) &= \exp(2\pi i \cdot  \al x / n)
\end{align*}
for $\al \in \Z_n$ with normalized Laplacian eigenvalue $\lam_\al = 1 - \cos(2\pi \al/n)$.
Estimating $1-\cos(2\pi \al /n) = \Theta(\al^2/n^2)$ for $\al = o(n)$
we obtain $\mul_{\phi} = \Theta(\sqrt{n})$.

In contrast we compute that $\CD_{\eps} = \Theta(1)$ for the cycle graph, which may be predicted since the eigenvector to $\lam_2$ already sign-indicates a sparsest cut.
For simplicity, assume $4 | n$ and let $Q^* = \{-\frac n4,-\frac n4 + 1,\dots, \frac n4 - 1\}$ be an optimal cut.
Let the eigenbasis representation of $Q^*$ be $\mathbf{1}_{Q^*} = \sum_{\al \in \Z_n} q_\al v_\al$.
We compute,
\begin{align*}
    \sum_{x \in Q^*} v_\al(x) &= \begin{cases}
        \Theta(\frac{n}{\al}) & \al \text{ odd}\\
        0 & \al \text{ even}
    \end{cases}\\
    \implies q_\al = \frac{\sum_{x \in Q^*} v_\al(x)}{\sum_{x \in \Z_n} |v_\al(x)|^2} &= \begin{cases}
        \Theta(\frac{1}{\al}) & \al \text{ odd}\\
        0 & \al \text{ even}
    \end{cases}
\end{align*}
Therefore, the squared weight on eigenspace $\al$ decreases at the rate $\Theta(\frac{1}{\al^2})$.
Summing this up, the total squared weight beyond the $1/\eps$-th eigenvalue is $\Theta(\eps)$.
At the same time, in the equation,
\begin{align*}
    \phi(Q^*) &= \E_{\al \sim \cS(Q^*)}[\lam_\al]\\
    &= \sum_{\al \in \Z_n} \Theta(\tfrac{1}{\al^2})\cdot \Theta(\al^2/n^2)\\
    &= \Theta(\tfrac 1n)
\end{align*}
each Fourier level is contributing the same amount $\Theta(\frac 1{n^2})$ towards the expansion of $Q^*$.

\subsection{Threshold rank upper bounds solution dimension}\label{sec:threshold-rank-vs-cut-dimensions}

We formally show here that the solution dimension is upper bounded by the $\psi$-threshold rank.

\begin{fact}
    Let $G$ be a graph and let $0<\eps<1\,.$
    Then $\CD_{\eps}(G)\leq \mul_{\psi(G)/\eps}(G)\,.$
    \salil{should the factor of 2 be in the subscript?  I assume that comes from moving between $\phi$ and $\psi$} \jiyu{updated the proof for sparsity and the parameter}
    \chris{updated the proof to include the missing centering of the indicator vector (and also to hold for irregular graphs)}

\begin{proof}
    Let $Q \subseteq V(G)$ be any sparsest cut. Let $v_1,\dots, v_n$ be the eigenbasis for $G$, and expand $\matD^{1/2}\mathbf{1}_Q$ in the eigenbasis as,
    \[
        \matD^{1/2}\mathbf{1}_Q = \sum_{i = 1}^n q_i v_i\,.
    \]
    The normalized density is, using \cref{fact:rayleigh-quotient},
    \[
        \psi_G(Q) = \psi(G) = \frac{\bar{\mathbf{1}}_Q^T \matD^{1/2}\matL \matD^{1/2}\bar{\mathbf{1}}_Q}{\bar{\mathbf{1}}_Q^T \matD \bar{\mathbf{1}}_Q} = \frac{\sum_{i=2}^n q_i^2 \lam_i}{\sum_{i=2}^n q_i^2}\,.
    \]
    The last quantity can be interpreted as the expected value of $\lam_i$ for $i \in [n] \setminus \{1\}$ drawn with probability proportional to $q_i^2$\,.
    Since $\lam_i \geq 0$ we can apply Markov's inequality.
    For all $\tau > 0$ we have:
    \[
        \frac{\sum_{\substack{i \in [n]:\\\lam_i \geq \tau}} q_i^2}{\sum_{i = 2}^n q_i^2} \leq \frac{\psi(G)}{\tau}\,.
    \]
    Taking $\tau = \psi(G) / \eps$\,, the right-hand side becomes $\eps$.
    Hence $\matD^{1/2}\bar{\mathbf{1}}_Q$ has all but $\eps$ fraction of its spectral mass on eigenvalues up to $\tau = \psi(G)/\eps$ which implies $\CD_\eps(G) \leq \mul_{\psi(G)/\eps}$\,.
\end{proof}
\end{fact}

\newcommand*{\horzbar}{\rule[.5ex]{2.5ex}{0.5pt}}
\newcommand{\RS}{\mathsf{RS}}
\newcommand{\had}{\mathsf{Had}}

\end{document}